\documentclass[11pt]{article}
\usepackage{fullpage}
\usepackage[bookmarks]{hyperref}
\usepackage{amssymb}
\usepackage{amsmath}
\usepackage{amsthm}
\usepackage{youngtab}
\usepackage{graphicx,color,colordvi}
\usepackage{subfig}
\usepackage{enumerate}
\usepackage{bbm}
\usepackage{cleveref}
\usepackage[utf8]{inputenc}
\usepackage{fullpage}
\usepackage[blocks]{authblk}
\usepackage{mathtools}
\usepackage{multirow}
\usepackage{indentfirst}

\newcommand{\nc}{\newcommand}
\nc{\rnc}{\renewcommand}
\nc\mnb[1]{\medskip\noindent{\bf #1}}

\newcommand{\M}{{\mathbb{M}}}

\newcommand{\ot}{\otimes}

\newcommand{\rank}{\operatorname{rank}}
\newcommand{\<}{\langle}

\newcommand{\Hom}{\operatorname{Hom}}
\newcommand{\Span}{\operatorname{span}}
\renewcommand{\>}{\rangle}
\newcommand{\ind}{\operatorname{ind}}
\newcommand\be{\begin{equation}}
\newcommand\ee{\end{equation}}
\newcommand{\A}{\mathcal{A}_{n}^{t_{n}}(d)}
\newcommand{\MM}{\mathcal{M}}

\DeclareMathOperator{\tr}{Tr}

\usepackage{hyperref}
\hypersetup{colorlinks, citecolor=magenta, filecolor=blue, linkcolor=blue, urlcolor=green}
\newtheorem{theorem}{Theorem}

\newtheorem{corollary}[theorem]{Corollary}

\newtheorem{definition}[theorem]{Definition}

\newtheorem{lemma}[theorem]{Lemma}
\newtheorem{notation}[theorem]{Notation}

\newtheorem{proposition}[theorem]{Proposition}
\newtheorem{remark}[theorem]{Remark}

\begin{document}
\title{Simplified formalism of the algebra of partially transposed permutation operators with applications}

\author{Marek Mozrzymas}
\affil[1]{\small Institute for Theoretical Physics, University of Wrocław
	50-204 Wrocław, Poland}
\author{Micha{\l} Studzi{\'n}ski}
\affil[1]{\small DAMTP, Centre for Mathematical Sciences, University of Cambridge, Cambridge~CB30WA, UK} 
\author{Micha{\l} Horodecki}
\affil[3]{\small Institute of Theoretical Physics and Astrophysics, National Quantum Information Centre, Faculty of Mathematics, Physics and Informatics, University of Gda{\'n}sk, Wita Stwosza 57, 80-308 Gda{\'n}sk, Poland}

\maketitle			 
\begin{abstract}
Hereunder we continue the study of the representation theory of the
algebra of permutation operators acting on the $n$-fold tensor product space, partially
transposed on the last subsystem. We develop the concept of partially reduced
irreducible representations, which allows to simplify significantly
previously proved theorems and what is the most important derive new results for irreducible representations of the mentioned algebra. In our analysis we  are able to reduce complexity of the central expressions by getting rid of sums over all permutations from symmetric group obtaining equations which are  much more handy in practical applications. We also find relatively simple matrix representations for the generators of underlying algebra.  Obtained simplifications and developments are applied to derive characteristic of the deterministic port-based teleportation scheme written purely in terms of irreducible representations of the studied algebra. We solve an eigenproblem for generators of algebra which is the first step towards to hybrid port-based teleportation scheme and gives us new proofs of asymptotic behaviour of teleportation fidelity. We also show connection between density operator characterising port-based teleportation and particular matrix composed of irreducible representation of the symmetric group which encodes properties of the investigated algebra. 
\end{abstract}

\section{Introduction}
\bigskip 
In the classical theory of representation of the symmetric group $S(n)$, the swap or permutation representation, closely related with the famous Schur-Weyl duality~\cite{Wallach} plays an important role not only in pure mathematics, but have found wide range of  applications in quantum information theory. We point here only a few of them, i.e. qubit quantum cloning~\cite{Qclon}, theory of quantum gates~\cite{Harrow1,Gulio3,Gulio4}, quantum error correcting codes~\cite{Junge}, distillation of quantum entanglement~\cite{Czech}, task of optimal compression of identical particles~\cite{Gulio2} or some aspects of theory of reference frames~\cite{Gulio}. The basic objects in the description are the operators $V_d(\sigma)$ representing the group elements $\sigma \in S(n)$ acting in the tensor product space $(\mathbb{C}^d)^{\ot n }$ as a permutations on the basis vectors in the algebra of tensor operators. It turns out that above mentioned picture can be modified introducing into it the notion of the partial transposition which has well established position in entanglement detection theory due to famous Peres-Horodecki criterion~\cite{HORODECKI19961,PhysRevLett.77.1413}, investigating set of PPT states~\cite{Tura}, or from the view of representation theory recent developments in port-based teleportation protocol~\cite{Stu2017,Moz2017b} and theory of universal qudits quantum cloning machines~\cite{Stu2014}, where representation theory of partially transposed operators $V_d(\sigma)$ was used. Because of this non-trivial connection between possible practical applications and pure mathematical theory, there is need to study much more deeper properties of the operators $V_d(\sigma)$ affected by partial transposition. In the series of papers~\cite{Stu1,Moz1} the first big step was done and the theory of irreducible representations (irreps) of partially transposed permutation operators $V_d^{t_n}(\sigma)$  was formulated, where by $t_n$ we denote partial transposition with respect to last subsystem.  It is known that there is a connection of algebra $\A$ with the  Walled Brauer Algebra (WBA)~\cite{Benkart1,Cox1,Koike1} which is a sub-algebra of the  Brauer Algebra~\cite{Brauer1,Pan1,Gavarini1}. Namely algebra $\A$ is a representation of WBA~\cite{Zhang1}. From~\cite{Brundan1} we know that whenever $d>n-1$ the dimension of WBA is equal to $n!$ which is the same as for $\A$~\cite{Stu1,Moz1}. In this case these two algebras are isomorphic and we know characterisation of the irreducible components. When condition $d>n-1$ is not satisfied we have $\operatorname{dim}A_{n}^{t_{n}}(d) <n!$, while dimension of WBA is still equal to $n!$. Because of that we do not have isomorphic between these two algebras and  the full investigation in this case was needed. There were even more to know. Namely in both cases form the point of view of possible applications additional knowledge about explicit orthogonal bases in every irreducible space and matrix representations of irreps is required. As we mentioned at the beginning of this section such full characterisation of $\A$ including matrix representations of irreps for all possible relations between $n$ and $d$ was presented in~\cite{Moz1,Stu1}. 

The formulas describing the representation theory of the algebra $\A$  in the original picture given in~\cite{Moz1,Stu1} are complicated and hard to use in practical applications.
 Fortunately  it turns out
that if we use particular kind of irreps of the symmetric
group, namely a partially reduced irreducible representation (PRIR), the
complicated expressions in representation theory of the algebra $%
\A$ may be simplified significantly and may be written in more
explicit way, much easier for their applications. The biggest profit is reduction  the complexity of existing equations by reducing number of sums over all permutations from symmetric group which significantly improves time of computations and allows us to prove new properties of $\A$. It should be mentioned
that applying PRIRs  in fact we do not loose generality, because any irrep of $S(N)$ may be unitarily transformed into PRIR form. Because of the high importance of described tools further analysis is required, so in this paper we develop the idea of PRIRs firstly introduced in~\cite{Stu2017} by presenting their new properties.

Our paper is organized in the following way. In Section~\ref{Sec3} we remind briefly
the structure of irreps of the algebra of partially
transposed permutation operators $\A$. We present all the most important for this manuscript theorems in their original form taken from~\cite{Moz1}.  In Section~\ref{AppA} we develop significantly the
concept of PRIR by deriving  its new properties. Next  in Section~\ref{simplification} we apply
newly derived results for PRIR to the existing formalism of irreps of algebra $\A$. In particular we simplify expressions for two special matrices describing properties of underlying algebra and projections onto irreducible subspaces of $\A$. Next
we deliver much more simpler matrix forms of the generators  of algebra $\A$.   Section~\ref{Appl_PBT} is fully devoted to the applications of the
simplified representation formalism of the algebra $\A$ to the
deterministic port-based teleportation. We show deep connection between port-based teleportation operator and matrix which encodes properties of investigated algebra.
We solve an eigenproblem for matrix generators of $\A$ which is one step forward to  hybrid scheme of PBT~\cite{ishizaka_remarks_2015}. This result allows us also to present an alternative proof of lower bound on teleportation fidelity presented in~\cite{beigi_konig}. In particular we derive in a simpler way than previous  expressions
for the fidelity in deterministic version of the protocol and we describe some of its asymptotic properties. 

\section{Representation theory of algebra of the partially transposed permutation operators}
\label{Sec3}
For self-consistence of the manuscript we present here all the most important  facts about algebra $\A$ of partially transposed permutation operators preceded by an introduction of the notation which is essential for proper understanding of further sections.  Next in the subsection~\ref{sec1} we only briefly summarize in possible simple way the structure of $\A$ and explain why original picture is inconvenient in practical use. This should to give to a reader the flavour of the problem before subsection~\ref{sec2}, where all important technical details are presented. In both following subsections as well later on we keep here original notation taken from~\cite{Stu1,Moz1}.

\subsection{Definitions and Notations}
Let us start here form considering a permutational representation $V$ of the group $S(n)$ in the space $\mathcal{H\equiv (\mathbb{C}}^{d})^{\otimes n}$ defined in the following way

\begin{definition}
	$V:S(n)$ $\rightarrow \Hom(\mathcal{(\mathbb{C}}^{d})^{\otimes n})$ and 
	\be
	\forall \sigma \in S(n)\qquad V(\sigma ).e_{i_{1}}\otimes e_{i_{2}}\otimes
	\cdots \otimes e_{i_{n}}=e_{i_{\sigma ^{-1}(1)}}\otimes e_{i_{\sigma
			^{-1}(2)}}\otimes \cdots \otimes e_{i_{\sigma ^{-1}(n)}},
	\ee
	where $d\in \mathbb{N}$ and $\{e_{i}\}_{i=1}^{d}$ is an orthonormal basis of the space $\mathcal{\mathbb{C}}^{d}.$ 
\end{definition}

The representation $V: S(n)$ $\rightarrow \Hom(\mathcal{(\mathbb{C}}^{d})^{\otimes n})$ is defined in a given basis $\{e_{i}\}_{i=1}^{d}$ of
the space $\mathcal{\mathbb{C}}^{d}$ (and consequently in a given basis of $\mathcal{H}$), so in fact it is a matrix
representation.

\begin{remark}
	The representation $V: S(n)$ $\rightarrow \Hom(\mathcal{(\mathbb{C}}^{d})^{\otimes n})$ which we will denote shortly as $V(S(n))$ depends explicitly on the dimension $d$, so in fact we
	should write $V(S(n))\equiv V_{d}(S(n))$ but for simplicity we will omit the index $d$,
	unless it will be necessary.
\end{remark}

Let us assume that we are given with the partition $\mu =(\mu _{1},\mu _{2},\ldots,\mu_{k})$ of some natural number $n$ (we denote this by $\mu \vdash n$), then by  $h(\mu)$ we denote the height (equivalently number of rows) of the corresponding Young diagram $Y^{\mu}$. Since there is one-to-one correspondence between partitions of natural number $n$  and inequivalent irreps of the symmetric group $S(n)$ we use symbols $\alpha, \mu$ etc. interchangeably for partitions and irreps  $\varphi^{\alpha},\psi^{\mu}$ whenever it is clear from the context or simplify notation. For $d>1$ the representation $V(S(n))$ is always reducible and we have

\begin{proposition}
	\label{mult_a}
	The irrep $\psi^{\mu}$ of $S(n),$ indexed by
	the partition $\mu \vdash n$ is
	contained in $V(S(n))$ if $d\geq k\equiv h(\mu).$ In particular if $d\geq n$
	then all irreps of $S(n)$ are included in the
	representation $V(S(n))$. When $d\geq k\equiv h(\mu)$ then the
	multiplicity of the irrep $\psi^{\mu}$ of $S(n)$
	is equal to%
	\be
	\frac{1}{n!}\sum_{\sigma \in S(n)}\chi ^{\mu}(\sigma ^{-1})d^{l(\sigma
		)}=\frac{1}{n!}\sum_{\sigma \in S(n)}\chi ^{\mu}(\sigma ^{-1})\chi
	^{V}(\sigma ),
	\ee
	where $\chi ^{\mu}(\cdot)$ is the character of $\psi ^{\mu}$, $l(\sigma
	) $ is the number of cycles in the permutation $\sigma $, and $\chi
	^{V}(\sigma )=d^{l(\sigma )}$ is the character of the representation $V(S(n))$.
\end{proposition}

The representation $V(S(n))$ extends in a natural way to the
representation of the group algebra $\mathbb{C}[S(n)]$ and in this way we get the algebra 
\be
\mathcal{A}_{n}(d)\equiv \Span_{\mathbb{C}}\{V(\sigma ):\sigma \in S(n)\}\subset \Hom(\mathcal{(\mathbb{C}}^{d})^{\otimes n}) 
\ee
of operators representing the elements of the group algebra $\mathbb{C}[S(n)]$.
Note that the algebra $A_{n}(d)$ contains a natural subalgebra 
\be
\label{An-1}
\mathcal{A}_{n-1}(d)\equiv \Span_{\mathbb{C}}\{V(\sigma _{n-1}):\sigma _{n-1}\in S(n-1)\}. 
\ee
The algebra of partially transposed operators with respect to last subsystem, the basic object for this manuscript is defined in the following way

\begin{definition}
	\label{def_A}
	For $\mathcal{A}_{n}(d)\equiv \Span_{\mathbb{C}}\{V(\sigma ):\sigma \in S(n)\}$ we define a new complex algebra 
	\be
	\mathcal{A}_{n}^{t_{n}}(d)\equiv \Span_{\mathbb{C}}\{V^{t_{n}}(\sigma ):\sigma \in S(n)\}\subset \Hom(\mathcal{(\mathbb{C}}^{d})^{\otimes n}),
	\ee
	where the symbol $t_{n}$ denotes the partial transposition with respect to the last subsystem 
	in the space $\Hom(\mathcal{(\mathbb{C}}^{d})^{\otimes n})$. The elements $V^{t_{n}}(\sigma ):\sigma \in S(n)$ will
	be called natural generators of the algebra $\mathcal{A}_{n}^{t_{n}}(d)$.
\end{definition}

Directly from Definition~\ref{def_A}, and equation~\eqref{An-1} it follows that $\mathcal{A}_{n-1}(d)\subset
\mathcal{A}_{n}^{t_{n}}(d)$.

\begin{notation}
	\label{notationV}
	Further in the text whenever partial transposition $t_n$ changes the elements $V(\sigma_n)\in V(S(n))\subset \mathcal{A}_{n}^{t_{n}}(d)$ we will write $V'(\sigma)$ instead of $V^{t_n}(\sigma)$. In particular whenever $\sigma=(n-1,n)$ we will write simply $V'$.
	When the partial transposition does not change the elements $V(\sigma _{n})$ $%
	\in $ $V(S(n-1))\subset \mathcal{A}_{n}^{t_{n}}(d)$ therefore, in the following, we
	will write  $V(\sigma _{n-1})$ instead of $V^{t_{n}}(\sigma _{n-1})$. 
\end{notation}

\begin{remark}
	\label{split}
	The algebra $A_{n}^{t_{n}}(d)$ \ is defined as the algebra of
	operators acting in the space $\mathcal{(\mathbb{C}}^{d})^{\otimes n}$, so in this way we get a natural representation of the
	algebra $A_{n}^{t_{n}}(d)$ in the space $\mathcal{(\mathbb{C}}^{d})^{\otimes n}$. The algebra $A_{n}^{t_{n}}(d)$ is semisimple~\cite{Moz1}, so
	this natural representation is a direct sum of irreps of the algebra $%
	A_{n}^{t_{n}}(d).$ 
\end{remark}

\subsection{Introduction to the problem}
\label{sec1}
The important feature of the algebra $
\A$, is the fact that it contains a subalgebra $\mathcal{A}_{n-1}(d)$,
generated by operators representing the subgroup $S(n-1)\subset S(n)$,
which are not changed by the partial transposition (these operators will
be denoted $V_{d}(\sigma ):\sigma \in S(n-1)$). In the mentioned papers~\cite{Stu1,Moz1} it has been shown that algebra $\A$ splits into direct sum of two left ideals $\A=\mathcal{M}\oplus \mathcal{S}$ which differ structurally and in consequence structure of the irreps of the algebra $\A$ is of two kinds and it is strictly connected with irreps of the groups $S(n-2)$ and $S(n-1)$. In particular the
matrix elements of the representations of the algebra $\A$ are
expressed in terms of matrix elements of the representations of the groups $
S(n-2)$ and $S(n-1)$.   Irreps of the first type  of $\A$ (called later non-trivial)  are indexed by irreps of the group $S(n-2)$ and they are strictly connected with
the representations of the group $S(n-1)$. Speaking more precisely in these representations, when the condition $%
d>n-2$ is satisfied, the elements $V_{d}(\sigma ):\sigma \in S(n)$ of
the algebra $\A$ are represented as in the representations of
the group $S(n-1)$ induced by irreps of $S(n-2)$ and
the dimension of such a representation of the first type of $\A
$ is equal to the dimension the induced representation of $S(n-1).$ When $%
d\leq n-2$ the situation is more complicated. In this case some of
the irreps of the first type may be defined on some
subspace of the representation space of induced representation of $S(n-1)$.
In both cases the non-trivially partially transposed generators  are
represented in these representations by complicated expression. In particular
the equation for transposition generators $V_{d}^{t_{n}}[(a,n)]$, where $a=1,\ldots, n-1$  as well as important expression for the projectors onto the
the non-trivial irreducible subspaces of the algebra $\A$
are also very complicated and have high complexity since we have to deal with sums over all elements from permutation group. These complicated
formulas were derived for arbitrary form of the irreps
of the groups $S(n-2)$ and $S(n-1)$ in terms of which the the
representations of the algebra $\A$ are expressed, so they are not really handy in terms of some applications. The representations of the second type are
indexed by some irreps of the group $S(n-1)$. In this case the
generators $V(\sigma ):\sigma \in S(n-1)$\ of the algebra $\A$ are represented naturally by operators of irreps of $%
S(n-1)$, whereas the non-trivially partially transposed operators, are
represented by zero operators. So in the representations of this 
type only the subalgebra $\mathcal{A}_{n-1}(d)$ of $\A$ is represented
non-trivially, therefore this irreps of the algebra $
A_{n}^{t_{n}}(d)$ may be called semi-trivial and we will not consider them.

\subsection{Technical summary of known results}
\label{sec2}
As we will see later in the analysis of the algebra $\A$ as well in applications to port-based teleportation very important
role plays the matrix $Q,$ which appears naturally in the theory of
representation of the algebra $\mathcal{A}_{n}^{t_{n}}(d)$, namely we have (see~\cite{Stu1},~\cite{Moz1})

\begin{definition}
	\label{def_Q}
For any irrep $\varphi ^{\alpha }$ of dimension $d_{\alpha}$ of the group $S(n-2)$
we define the block matrix 
\be
Q_{n-1}^{d}(\alpha )\equiv Q(\alpha )=(d^{\delta _{ab}}\varphi _{ij}^{\alpha
}[(a,n-1)(a,b)(b,n-1)])=(Q_{ij}^{ab}(\alpha ))\in \M((n-1)d_{\alpha },\mathbb{C}),
\ee
where $a,b=1,\ldots,n-1,\quad i,j=1,\ldots,d_{\alpha }$. The blocks
of the matrix $Q(\alpha )$ are labelled by indices $(a,b)$, whereas the
elements of the blocks are labelled by the indices of the irreducible
representation $\varphi ^{\alpha }=(\varphi _{ij}^{\alpha })$ of the group $%
S(n-2)$.
\end{definition}

Below we recall the most important spectral properties of above defined matrices.

\begin{proposition}
The matrices $Q(\alpha )$ are hermitian, positive semi-definite. Eigenvalues 
$\lambda _{\nu }(\alpha )$ of $Q(\alpha)$ are labelled by the irreps $\psi ^{\nu }\in \ind_{S(n-2)}^{S(n-1)}(\varphi ^{\alpha })$, and the
multiplicities of $\lambda _{\nu }(\alpha )$ are equal to $d_{\nu }$.
Moreover at most one (up to the multiplicity) eigenvalue $\lambda _{\nu }(\alpha)$
of the matrix $Q(\alpha )$ may be equal to zero.
\end{proposition}

\begin{remark}
The matrix $Q_{ij}^{ab}(\alpha)$ in the representation space has a form:
\be
\label{formQ}
Q_{ij}^{ab}(\alpha)=\begin{pmatrix}d\mathbf{1} & \varphi^{\alpha}[(1,2)] & \ldots & \varphi^{\alpha}[(1,n-2)] & \mathbf{1} \\ \varphi^{\alpha}[(2,1)] & d\mathbf{1} & \ldots & \varphi^{\alpha}[(2,n-2)] & \mathbf{1} \\ \vdots & &\ddots & & \vdots\\ \mathbf{1}& & \ldots & & d\mathbf{1} \end{pmatrix},
\ee
where every $\varphi^{\alpha}[(a,b)]=\{\varphi_{ij}^{\alpha}[(a,b)]\}$ is a representation matrix of permutation $(a,b)$ in irrep of $S(n-2)$ labelled by $\alpha$. It is worth to mention here that in general case there is always possibility to chose matrices $\varphi^{\alpha}$ to be unitary, so we get $\varphi^{\alpha}_{ij}[(a,b)]=\bar{\varphi}^{\alpha}_{ji}[(a,b)]$. In our paper constrains are even stronger because representations $\varphi^{\alpha}[(a,b)]$ are in the form of symmetric and real matrices, so we have $\varphi^{\alpha}_{ij}[(a,b)]=\varphi^{\alpha}_{ji}[(a,b)]$. 
\end{remark}
Before we go further let us define a few additional concepts in particular so called rank of the partition for corresponding Young diagram.
\begin{definition}
By $\widehat{S}(N)$ we denote the set of all irreps of the symmetric group $S(N)$, and by $|\widehat{S}(N)|$ its cardinality.
If $\psi^{\mu} \in \widehat{S}(N)$ is irrep of the group $S(N)$ we write
\be
\widehat{S}_{d}(N)\equiv\{\psi^{\mu} \in \widehat{S}(N) :  h(\mu )\leq d\}\Rightarrow 
\widehat{S}_{N}(N)=\widehat{S}(N).
\ee
Above set contains all irreps of $S(n)$ whose corresponding Young diagrams have no more than $d$ rows. 
\end{definition}
\begin{definition}
	\label{def_r}
 Let $\psi^{\mu}$ be any irrep of the group $
S(n)$, $\mu =(\mu_{1},\ldots,\mu_{k})$ its partition, and $
Y^{\mu}$ the corresponding Young diagram. The rank $r=r(\mu)$ (or 
$r\left(Y^{\mu }\right)$) of the partition $\mu$ is the length of the diagonal
of its Young diagram.
\end{definition}
Now we are in position to present spectral theorem  for the matrices $Q(\alpha)$ (see~\cite{Moz1}).
\begin{theorem}
	\label{thm18}
	\begin{enumerate}[a)]
		\item Let $\varphi^{\alpha}$ be any irrep of the group $
		S(n-2)$, $\alpha =(\alpha _{1},\ldots,\alpha_{k})$ its partition and $
		Y^{\alpha}$ the corresponding Young diagram. Suppose that for some index $1\leq i \leq k$ the sequence $\nu=(\alpha
		_{1},\ldots,\alpha_{i}+1,\ldots,\alpha_{k})$ is a partition of $n-1$, so it defines
		an irrep $\psi^{\nu}$ of the group $S(n-1)$. For
		 Young diagrams it means that the Young diagram $Y^{\nu}$ is obtained
		from the Young diagram $Y^{\alpha}$ by adding, in the $i$-th row, one box (we denote this by $\nu \in \alpha$).
		Then the corresponding matrix $Q_{n-1}(\alpha)$ has the following eigenvalues
		
		\begin{enumerate}[i)]
			\item if $r\left(Y^{\alpha}\right)=r\left(Y^{\nu}\right)$, then 
			\be
			\lambda_{\nu}(\alpha)=d+\alpha_{i}+1-i,\quad \quad i=1,\ldots,k+1,
			\ee
			and if $i=k+1$ we set $\alpha_{k+1}=0$.
			
			\item If $r\left(Y^{\alpha}\right)+1=r\left(Y^{\nu}\right)$ which may occur only if $i=r+1$, then 
			\be
			\lambda _{\nu}(\alpha)=d.
			\ee
		\end{enumerate}
		The case $ii)$ describes the situation when adding, in a proper way one box
		to Young diagram $Y^{\alpha}$ we extend its diagonal. The multiplicity of
		the eigenvalue $\lambda_{\nu}(\alpha)$ is equal to $\dim\psi^{\nu}$, and the
		number of pairwise distinct eigenvalues of the matrix $Q_{n-1}(\alpha)$ is
		equal to the rank of Young diagrams $Y^{\nu}$ that one can obtain from the Young
		diagram $Y^{\alpha}$ by adding, in a proper way, one box.
		\item  The unitary matrix $Z(\alpha)=\left( z(\alpha\right)_{kj_{\nu }}^{a\nu})$
		which reduces the induced representation $\Phi^{\alpha
		}=\operatorname{ind}_{S(n-2)}^{S(n-1)}(\varphi^{\alpha})$ into the irreducible components
		has a form 
		\be
		\label{zzz}
		z(\alpha )_{kj_{\nu }}^{a\nu }=\frac{d_{\nu }}{\sqrt{%
				N_{j_{\nu }'}^{\nu }}(n-1)!}\sum_{\sigma \in S(n-1)}\psi _{j_{\nu }'j_{\nu
		}}^{\nu }\left( \sigma ^{-1}\right) \delta _{a\sigma (q')}\varphi _{kr'}^{\alpha
		}[(an-1)\sigma (q'n-1)],
		\ee
		with 
		\be
		\label{NN}
		N_{j_{\nu }'}^{\nu }\equiv \left( E_{j_{\nu }'j_{\nu }'}^{\nu
		}\right)_{r'r'}^{q'q'}=\frac{d_{\nu }}{(n-1)!}\sum_{\sigma \in S(n-1)}\psi
		_{j_{\nu }'j_{\nu }'}^{\nu }\left(\sigma ^{-1}\right) \delta _{q'\sigma (q')}\varphi
		_{r'r'}^{\alpha }[(q'n-1)\sigma (q'n-1)],
		\ee
		where $\psi^{\nu}$ are representations of the group $S(n-1)$ whose Young
		diagrams are obtained from the Young diagram $Y^{\alpha}$ by adding, in a proper way, one box and $\left(\psi_{j_{\nu}j_{\nu}}^{\nu
		}(\sigma)\right)$ is a matrix form of $\sigma \in S(n-1)$ in the representation $
		\psi^{\nu}, E_{j_{\nu}'j_{\nu}'}^{\nu}$ is a
		hermitian projector of rank one in the representation space $\Phi^{\alpha}$
		defined by $\psi^{\nu}$ (see Def. 79 in App. C of~\cite{Moz1} ). The double index $
		(q',r')$ is fixed and chosen in such a way that $N_{j_{\nu}'}^{\nu}>0$,
		which  is always possible because $E^{\nu}_{j_{\nu}'j_{\nu}'}$ is a positive semi-definite matrix. Moreover we have 
		\be
		\sum_{ak}\sum_{bl}z^{\dagger}(\alpha)_{j_{\xi}k}^{\xi a}\Phi ^{\alpha}(\sigma
		)_{kl}^{ab}z(\alpha)_{lj_{\mu}}^{b\mu}=\delta^{\xi \mu}\psi_{j_{\xi
			}j_{\mu}}^{\mu}(\sigma).
		\ee
		In particular 
		\be
		\sum_{ak}\sum_{bl}z^{\dagger}(\alpha)_{j_{\xi}k}^{\xi a}Q(\alpha
		)_{kl}^{ab}z(\alpha)_{lj_{\mu }}^{b\mu }=\delta ^{\xi \mu}\delta
		_{j_{\xi}j_{\mu}}\lambda _{\mu}(\alpha),
		\ee
		so the columns of the matrix $Z(\alpha)=\left(z(\alpha )_{kj_{\nu }}^{a\nu }\right)$
		are eigenvectors of the matrix $Q(\alpha)$.
	\end{enumerate}
\end{theorem}

\begin{remark}
	\label{Hard}
	The indices $j_{\nu }^{\prime },$ $q^{\prime },r^{\prime }$, $%
	a^{\prime }$ are in fact parameters. Expression~\eqref{zzz} is  complicated in the practical applications because, although it looks quite explicit,  the normalisation factor $%
	N_{j_{\nu }^{\prime }}^{\nu }\equiv (E_{j_{\nu }^{\prime }j_{\nu }^{\prime
	}}^{\nu })_{r^{\prime }r^{\prime }}^{q^{\prime }q^{\prime }}$\ cannot be
	given explicitly without specifying the form of the irreps $%
	\psi ^{\nu }$ of $S(n-1)$ and irrep $\varphi ^{\alpha }$ of $S(n-2).$ 
Moreover equation~\eqref{zzz} contains two sums over all permutations from $S(n-1)$ which causes high complexity during explicit calculations.
\end{remark}

\begin{remark}
	The part a) of Theorem~\ref{thm18} gives an explicit 
	eigenvalues of the matrix $Q_{n-1}(\alpha)$ on the partition $\alpha
	=\left(\alpha_{1},\ldots,\alpha_{k}\right)$ which defines the the irreducible
	representation $\varphi^{\alpha}$ and  on the dimension parameter $d$.
\end{remark}
There is also another expression for eigenvalues $\lambda_{\nu}(\alpha)$ of matrices $Q(\alpha)$.
\begin{lemma}
	\label{semAn}
	Let $\varphi ^{\alpha }$ be any irrep of the group $%
	S(n-2), $ $\alpha $ its partition, $\chi^{\alpha }$ its character, and let $\psi ^{\nu }$ be all irreps of the group $S(n-1)$ whose Young diagrams are obtained
	from the Young diagram $Y^{\alpha} $ by adding, in a
	proper way, one box. By $\chi ^{\nu }$ we denote their characters,
	where $\nu $ is the partition of $n-1$ which labels the representation $\psi
	^{\nu }$. Then the distinct eigenvalues of the matrix $Q(\alpha )$ generated
	by the irrep $\varphi ^{\alpha }$ of $S(n-2)$ are
	labelled by the partitions $\nu $ and are of the form
	\be
	\label{seman}
	\lambda _{\nu }(\alpha)=d+\frac{(n-1)(n-2)}{2}\frac{\chi^{\nu }[(a,b)]}{d_{\nu }}-\frac{(n-2)(n-3)}{2}\frac{\chi^{\alpha }[(c,d)]}{d_{\alpha }}, 
	\ee
	where $(a,b)$ for $a,b\leq n-2$ is an arbitrary transposition in $S(n-2)$ and $(c,d)$ for $c,d\leq n-1$ is transposition in $S(n-1)$. The
	eigenvalue $\lambda _{\nu }(\alpha)$ has multiplicity $d_{\nu }$.
\end{lemma}

\begin{remark}
	Since irreducible characters $\chi^{\nu}[(a,b)],\chi^{\alpha}[(c,d)]$  are constant on conjugacy classes it is enough to take $(a,b)=(c,d)=(1,2)$ in equation~\eqref{seman} of Lemma~\ref{semAn}.
	Reader notices that the quantity $\lambda_{\nu}(\alpha)$ is of  non-zero value only if $\nu \in \alpha$. In this manuscript we assume that this assumption is always satisfied.
\end{remark}

In the next and last part of this section we briefly recall the basic properties of irreps  of the algebra $\mathcal{A}_{n}^{t_{n}}(d)$. The irreps of the
algebra $\mathcal{A}_{n}^{t_{n}}(d)$ are of two kinds and we describe them in the
matrix form. Let us start from the following

\begin{proposition}
	\label{f_basis}
The first kind of irreps, denoted by $\Phi ^{\alpha },$ are
determined by irreps $\varphi ^{\alpha }$ of the group $S(n-2)$, such
that $\varphi^{\alpha}\in \widehat{S}_d(n-2)$, and we have 
\be
\Phi ^{\alpha }:\mathcal{A}_{n}^{t_{n}}(d)\rightarrow \M(\rank Q(\alpha ),\mathbb{C}).
\ee
The representation space $S(\Phi ^{\alpha })$ of $\Phi ^{\alpha }$ has
the following structure 
\be
S(\Phi ^{\alpha })=\bigoplus _{\nu \in I:h(\nu )\leq d}S(\psi ^{\nu }), 
\ee
where 
\be
\ind_{S(n-2)}^{S(n-1)}(\varphi ^{\alpha })=\bigoplus _{\nu \in I}\psi ^{\nu }. 
\ee
In the above by the $I$ we denote the set of irreps of the group $S(n-1)$, which appear
in the above decomposition of induced representation of $S(n-1)$ into
irreducible components $\psi ^{\nu }$.  In the reduced matrix basis $%
f\equiv\{f_{j_{\nu} }^{\nu }:h(\nu )\leq d,\quad j_{\nu} =1,\ldots,d_{\nu }\}$ of
the representation $\Phi ^{\alpha }$ (see~\cite{Moz1}) the natural generators $%
V'[(a,n)]$ and $V(\sigma _{n-1})$ of $\mathcal{A}_{n}^{t_{n}}(d)$ are represented
by the following matrices%
\be
M_{f}^{\alpha }\left[ V'[(a,n)]\right] _{j_{\xi }j_{\nu }}^{\xi \nu
}=\sum_{k=1,\ldots,d_{\alpha }}\sqrt{\lambda _{\xi }}z^{\dagger}(\alpha
)_{j_{\xi }k}^{\xi a}z(\alpha )_{kj_{\nu }}^{a\nu }\sqrt{\lambda _{\nu }}%
:\xi ,\nu \in I,\;a=1,\ldots,n-1, 
\ee
\be
\label{expr1a}
M_{f}^{\alpha }\left[ V(\sigma _{n-1})\right] _{j_{\nu'}j_{\nu }}^{\nu'\nu }=\delta ^{\nu'\nu }\psi _{j_{\nu'}j_{\nu }}^{\nu
}(\sigma _{n-1}),\quad \sigma _{n-1}\in S(n-1), 
\ee
where the matrices $Z(\alpha )=(z(\alpha )_{kj_{\nu }}^{a\nu })$ are defined
in the Theorem~\ref{thm18}. Expression~\eqref{expr1a} shows that irrep $\Phi^{\alpha}$ of algebra $\A$ is $\ind_{S(n-2)}^{S(n-1)}(\varphi^{\alpha})$ of subalgebra $\mathbb{C}[S(n-1)]$.
\end{proposition}

The structure of the irreps of the second kind is much simpler,

\begin{proposition}
The irreps of the second kind, denoted as $\Psi ^{\nu }$, are
determined by the irreps $\psi^{\nu }$ of the group $S(n-1)$, such that $%
h(\nu )<d$. In this case we have 
\be
\Psi ^{\nu }:\mathcal{A}_{n}^{t_{n}}(d)\rightarrow \M(d_{\nu },\mathbb{C}), 
\ee
where the representation space $S(\Psi ^{\nu })$ of $\Psi ^{\nu }$ is simply 
$S(\Psi ^{\nu })=S(\psi ^{\nu })$, and 
\be
\Psi^{\nu }(a)=
\begin{cases}
0:a\notin S(n-1), \\ 
\psi ^{\nu }(\sigma _{n-1}):a=\sigma _{n-1}\in S(n-1).
\end{cases}%
\ee
In this case only the
elements of $S(n-1)$, which are not changed by partial transpose are
represented non trivially. The remaining natural generators of the algebra $%
\mathcal{A}_{n}^{t_{n}}(d)$ are not invertible~\cite{Moz1}.
\end{proposition}

Using the properties of irreps of $A_{n}^{t_{n}}(d)$ described in~\cite{Moz1} one
can derive the following decomposition of the natural representation of the
algebra $\mathcal{A}_{n}^{t_{n}}(d)$ (see Remark~\ref{split}) into irreducible components.

\begin{theorem}
	\label{decompA}
The algebra $\mathcal{A}_{n}^{t_{n}}(d)$ in its natural representation in the space $(\mathbb{C}^{d})^{\otimes n}$ has the following decomposition into irreps
\be
\label{decompA_eq}
\mathcal{A}_{n}^{t_{n}}(d)=\left[ \bigoplus _{\alpha :h(\alpha )\leq d}m_{\alpha }\Phi ^{\alpha
}\right] \oplus \left[  \bigoplus _{\nu :h(\nu )<d}m_{\nu }\Psi ^{\nu }\right].
\ee
The multiplicities $m_{\alpha }$ are equal to the multiplicities of the irreps $\varphi ^{\alpha }$ of $S(n-2)$ in the representation $V_{d}(S(n-2))$ (see Proposition~\ref{mult_a}) 
\be
m_{\alpha }=\frac{1}{(n-2)!}\sum_{\sigma \in S(n-2)}\chi ^{\alpha }(\sigma
^{-1})d^{l(\sigma )}, 
\ee
 where $l(\sigma
 ) $ is the number of disjoint cycles in the permutation $\sigma $, and 
\be
m_{\nu }=d^{n}-\sum_{\alpha :\nu \in \ind_{S(n-2)}^{S(n-1)}(\varphi ^{\alpha
})}m_{\alpha }. 
\ee
\end{theorem}

\begin{remark}
Note that from Theorem~\ref{decompA} it follows that when $d\geq n$ all possible 
irreps of the first kind and second kind are included in the
decomposition of $\mathcal{A}_{n}^{t_{n}}(d)$. When $d<n$ then the conditions $h(\alpha )\leq d$ and $h(\nu )<d$ limit the variety of irreps appearing in the decomposition given through~\eqref{decompA_eq}.
\end{remark}
As was presented in~\cite{Moz1}   the orthogonal  projectors  $F_{\nu }(\alpha )$ onto non-trivial irreducible subspaces of the algebra $\mathcal{A}_n^{t_n}(d)$ have in fact very complicated form:
\begin{proposition}
	\label{complicated}
	Projectors onto non-trivial irreducible spaces of algebra $\A$ are of the form
	\be
	\label{F_comp}
	F_{\nu }(\alpha )=\frac{1}{\lambda _{\nu}(\alpha)}\sum_{j_{\nu
	}}\sum_{ai,bk}(z^{-1}(\alpha))_{j_{\nu }i}^{\nu a}u_{ki}^{ba}(\alpha )z_{kj_{\nu
	}}^{b\nu }(\alpha),
	\ee
	where $z(\alpha)_{j_{\nu }i}^{\nu a}$ is given in eq.~\eqref{zzz} in Theorem~\ref{thm18}, and
	\be
	u_{ij}^{ab}(\alpha )=\frac{d_{\alpha }}{(n-2)!}V'[(a,n)]\sum_{\sigma \in
		S(n-2)}\varphi _{ji}^{\alpha }(\sigma ^{-1})V[(an-1)(\sigma )(bn-1)].
	\ee
\end{proposition}
Indeed, equation~\eqref{F_comp} contains matrix elements of $Z(\alpha)$ which are in generally hard to compute (see Remark~\ref{Hard}).

\section{New results regarding partially reduced irreducible representations}
\label{AppA}
In this note we recall the notion of  the Partially Reduced Irreducible Representations (PRIRs) introduced in~\cite{Stu2017} and derive new properties of them. The  concept of PRIRs plays a crucial role in the simplification of the representation of the algebra $\mathcal{A}_n^{t_n}(d)$ (see Section~\ref{simplification}), derivation of the fidelity in the deterministic version of the PBT, and new proof of the lower bound on fidelity in deterministic version of the PBT presented~\cite{beigi_konig} (see Section~\ref{det_PBT}).

Let us consider an arbitrary unitary irrep $\psi^{\mu }$ of $S(n)$. It
can be always unitarily transformed to reduced form $\psi_{R}^{\mu }$, such that 
\be
\forall \pi \in S(n-1)\quad \psi _{R}^{\mu }(\pi )=\bigoplus _{\alpha \in \mu
}\varphi ^{\alpha }(\pi), 
\ee
where $\varphi ^{\alpha }$ are irreps of $S(n-1)$. By $\alpha \in \mu$ we understand such Young diagrams $\alpha$ which can be obtained from $\mu$ by removing one box in the proper way. We see, that the restriction of the irrep  $\psi^{\mu }$ of $S(n)$ to the subgroup $S(n-1)$ has a block-diagonal form of completely reduced representation, which in matrix notation takes the form%
\be
\label{prir1}
\forall \pi \in S(n-1)\quad \psi _{R}^{\mu }(\pi)=\left( \delta ^{\alpha \beta
}\varphi _{i_{\alpha }j_{\alpha }}^{\alpha }\right) .
\ee
The block structure of this reduced representation allows us to introduce such a block indexation for $\psi _{R}^{\mu }$ of $S(n)$, which gives 
\be
\label{eq27}
\forall \sigma \in S(n)\quad \psi _{R}^{\mu }(\sigma )=\left( \psi _{k_{\mu
	}l_{\mu }}(\sigma )\right) =\left( \psi _{i_{\alpha }j_{\beta }}^{\alpha \beta }(\sigma)\right) ,
\ee
where the indices $k_{\mu},l_{\mu}$ are standard matrix indices, the matrices on the diagonal $(\psi _{R}^{\mu })^{\alpha \alpha
}(\sigma )=\left( \psi _{i_{\alpha }j_{\alpha }}^{\alpha \alpha }(\sigma )\right) $ are
of dimension of corresponding irrep $\varphi ^{\alpha }$ of $S(n-1)$. Reader notices that the
off-diagonal blocks need not to be square.  From this it follows that we may introduce the idea PRIR which we define in the following way
\begin{definition}
	An irrep $\psi ^{\mu}$of the group $S(n)$ is
	the Partially Reduced Irreducible Representation (PRIR) \ if it has a
	reduced form on the subgroup $S(n-1)\subset S(n),$ i.e. we have 
	\be
	\forall \sigma \in S(n-1)\quad \psi _{R}^{\mu }(\sigma )=\bigoplus _{\alpha \in \mu
	}\varphi ^{\alpha }(\sigma ).
	\ee
	For such  representations, in general,  we will use the block indexation
	described in equation~\eqref{eq27}.
\end{definition}

\begin{remark}
	Clearly for a given irrep $\psi ^{\mu}$ of the
	group $S(n)$ its PRIR is given not uniquely.
\end{remark}

\begin{remark}
It is obvious that any irrep of $S(n)$ can be unitarily transformed into PRIR representation.
\end{remark}
The first new result regarding PRIRs is summarized in the following proposition which plays similar role to standard orthogonality relation for irreps:
\begin{proposition}
	\label{AP13}
	The PRIRs $\psi _{R}^{\mu }, \psi _{R}^{\nu }$ of $S(n)$ satisfy the following bilinear summation rule
	\be
	\label{eqAP13}
	\forall \alpha ,\beta \in \mu  \quad \forall \beta, \gamma \in \nu \qquad \sum_{a=1}^{n}\sum_{k_{\beta
		}=1}^{d_{\beta}}(\psi _{R}^{\mu })_{i_{\alpha }k_{\beta }}^{\alpha \beta
	}[(a,n)](\psi _{R}^{\nu })_{k_{\beta }j_{\gamma }}^{\beta \gamma }[(a,n)]=n\frac{d_{\beta }}{d_{\mu }}\delta^{\mu \nu}\delta ^{\alpha \gamma }\delta
	_{i_{\alpha }j_{\gamma }}, 
	\ee
	where $\alpha, \beta, \gamma $ are irreps of $S(n-1)$ contained
	in the irreps $\mu,\nu $ of $S(n)$.
\end{proposition}
The proof of above proposition goes similarly as proof of the Proposition 17 in~\cite{Stu2017}, but it generalisation is necessary for further applications in this manuscript. 
Next we  we prove one more summation rule which is crucial in order to prove Theorem~\ref{BL16} which is the main result of this section.
\begin{lemma}
	\label{PLemm}
	\bigskip Let $\psi _{R}^{\nu }$ are PRIR representations of the group $%
	S(n-1)$ included in $\Phi ^{\alpha }=\ind_{S(n-2)}^{S(n-1)}(\varphi ^{\alpha })$, then
	\be
	\sum_{\nu \in \Phi^{\alpha }}d_{\nu} (\psi _{R}^{\nu })_{j_{\alpha }k_{\alpha
	}}^{\alpha \alpha }[(b,n-1)]=(n-1)d_{\alpha} \delta _{b,n-1}\delta _{j_{\alpha
		}k_{\alpha }},
	\ee
	where the summation is over partitions $\nu$ labelling irreps of $S(n-1)$  contained in $\Phi^{\alpha}$.
\end{lemma} 

\begin{proof}
	Let 
	\be
	\label{AAA}
	\sum_{\nu \in \Phi^{\alpha }}d_{\nu} (\psi _{R}^{\nu })_{j_{\alpha }k_{\alpha
	}}^{\alpha \alpha }[(b,n-1)]=x_{j_{\alpha }k_{\alpha }}(b),\qquad
	b=1,\ldots,n-1.
	\ee
	We define a hermitian matrix%
	\be
	X(b)\equiv \left( x_{j_{\alpha }k_{\alpha }}(b)\right) \in \mathbb{M}(\dim \varphi^{\alpha} ,\mathbb{C}),
	\ee
	such that 
	\be
	X(n-1)=(n-1)d_{\alpha}\mathbf{1}_{d_{\alpha}}, 
	\ee
	where $\mathbf{1}_{d_{\alpha}}$ denotes identity operator of dimension $d_{\alpha}$.
	Now from equation~\eqref{AAA} we get 
	\be
	\label{LHS}
	\sum_{\nu \in \Phi^{\alpha }}\sum_{k_{\alpha }}d_{\nu}(\psi _{R}^{\nu
	})_{j_{\alpha }k_{\alpha }}^{\alpha \alpha }[(b,n-1)](\psi _{R}^{\mu
	})_{k_{\alpha }l_{\alpha }}^{\alpha \alpha }[(b,n-1)]=\sum_{k_{\alpha
	}}x_{j_{\alpha }k_{\alpha }}(b)(\psi _{R}^{\mu })_{k_{\alpha }l_{\alpha
	}}^{\alpha \alpha }[(b,n-1)].
	\ee
	Making the summation over $b=1,\ldots,n-1$ and applying Proposition~\ref{AP13} to LHS of~\eqref{LHS}, we get 
	\be
	\label{LHS2}
	(n-1)d_{\alpha}\delta _{j_{\alpha }l_{\alpha }}=\sum_{k_{\alpha
	}}\sum_{b=1}^{n-1}x_{j_{\alpha }k_{\alpha }}(b)(\psi _{R}^{\mu })_{k_{\alpha
		}l_{\alpha }}^{\alpha \alpha }[(b,n-1)]. 
	\ee
	Multiplying both sides of~\eqref{LHS2} by $d_{\mu}$ and making the summation over $\mu \in
	\Phi^{\alpha }$ we have
	\be
	\begin{split}
		(&n-1)^{2}d_{\alpha}^{2}\delta _{j_{\alpha }l_{\alpha }}=\sum_{k_{\alpha
		}}\sum_{b=1}^{n-1}x_{j_{\alpha }k_{\alpha }}(b)x_{k_{\alpha l\alpha
		}}(b)=\sum_{b=1}^{n-1}x_{j_{\alpha }l_{\alpha }}^{2}(b)\\ 
		&=\sum_{b=1}^{n-2}x_{j_{\alpha }l_{\alpha }}^{2}(b)+x_{j_{\alpha }l_{\alpha
		}}^{2}(n-1)=\sum_{b=1}^{n-2}x_{j_{\alpha }l_{\alpha
		}}^{2}(b)+(n-1)^{2}d_{\alpha}^{2}\delta _{j_{\alpha }l_{\alpha }}, 
	\end{split}
	\ee
	which means that 
	\be
	\sum_{b=1}^{n-2}x_{j_{\alpha }l_{\alpha
	}}^{2}(b)=\sum_{b=1}^{n-2}\left( X^{2}(b)\right) _{j_{\alpha }l_{\alpha }}=0. 
	\ee
	From the above it follows that 
	\be
	\forall b=1,\ldots,n-2\,\,\qquad X^{2}(b)=0\Leftrightarrow X(b)=0, 
	\ee
	since the the matrices $X(b)$ are hermitian, so the matrices $X^{2}(b)$
	are positive semi-definite.
\end{proof}

From Lemma~\ref{PLemm} one can easily deduce the following

\begin{corollary}
	\label{PCor4}
	Let 
	\be
	\sigma =\gamma (b,n-1)\in S(n-1):\gamma \in S(n-2),\quad b=1,\ldots,n-1, 
	\ee
	then 
	\be
	\sum_{\nu \in \Phi^{\alpha }}d_{\nu}(\psi _{R}^{\nu })_{j_{\alpha }k_{\alpha
	}}^{\alpha \alpha }[\gamma (b,n-1)]=(n-1)d_{\alpha}\delta _{b,n-1}\varphi
	_{j_{\alpha }k_{\alpha }}^{\alpha }(\gamma ). 
	\ee
	In particular if $\sigma \in S(n-1)$, and $\sigma \notin S(n-2)\subset S(n-1)$, then
	\be
	\sum_{\nu \in \Phi^{\alpha }}d_{\nu}(\psi _{R}^{\nu })_{j_{\alpha }k_{\alpha
	}}^{\alpha \alpha }(\sigma )=0.
	\ee
	If $\sigma\in S(n-2)\subset S(n-1)$, then 
	\be
	\sum_{\nu \in \Phi^{\alpha }}d_{\nu}(\psi _{R}^{\nu })_{j_{\alpha }k_{\alpha
	}}^{\alpha \alpha }(\sigma)=(n-1)d_{\alpha}\varphi _{j_{\alpha }k_{\alpha
	}}^{\alpha }(\sigma ). 
	\ee
\end{corollary}
Now we are in the position to prove the main result of this section, namely we have the following
\begin{theorem}
	\label{BL16}
	Let $\psi _{R}^{\nu }$ are PRIR representations of the group $%
	S(n-1)$ included in $\Phi ^{\alpha }=\ind_{S(n-2)}^{S(n-1)}\left( \varphi ^{\alpha }\right) $ i.e.
	\be
	\Phi^{\alpha }=\ind_{S(n-2)}^{S(n-1)}(\varphi^{\alpha} )\simeq \bigoplus _{\nu \in \Phi^{\alpha }}\psi _{R}^{\nu }, 
	\ee
	then $\forall \sigma \in S(n-1)$ we have the following summation rule 
	\be
	\begin{split}
		\sum_{\nu \in \Phi^{\alpha }}d_{\nu} (\psi
		_{R}^{\nu })_{j_{\alpha }k_{\alpha }}^{\alpha \alpha }[(a,n-1)\sigma
		(b,n-1)]=(n-1)d_{\alpha}\delta _{a\sigma (b)}\varphi _{j_{\alpha }k_{\alpha
		}}^{\alpha }[(a,n-1)\sigma (b,n-1)].
	\end{split}
	\ee
	In particular we have 
	\be
	\sum_{\nu \in \Phi (\alpha )}d_{\nu}(\psi _{R}^{\nu })_{j_{\alpha }k_{\alpha
	}}^{\alpha \alpha }[(a,n-1)(b,n-1)]=(n-1)d_{\alpha}\delta _{ab}\delta
	_{j_{\alpha }k_{\alpha }}. 
	\ee
\end{theorem}

\begin{proof}
	From  Corollary~\ref{PCor4} it follows that, in order to calculate 
	\be
	\label{mianequality}
	\forall \sigma \in S(n-1)\quad \sum_{\nu \in \Phi^{\alpha }}d_{\nu}(\psi
	_{R}^{\nu })_{j_{\alpha }k_{\alpha }}^{\alpha \alpha }[(a,n-1)\sigma
	(b,n-1)], 
	\ee
	we need to establish when the permutation $(a,n-1)\sigma (b,n-1)$ belongs to
	the subgroup $S(n-2)\subset S(n-1),$ otherwise the sum is equal to zero. It
	is easy to check that 
	\be
	(a,n-1)\sigma (b,n-1)\in S(n-2)\subset S(n-1) 
	\ee
	if and only if 
	\be
	\sigma :b\mapsto a 
	\ee
	and from this it follows the statement of the theorem.
\end{proof}

\begin{proposition}
	\label{propoopo}
	Suppose that $\psi _{R}^{\nu }$ is a PRIR representation of the
	group $S(n)$, then we have 
	\be
	\sum_{a=1}^{n}(\psi _{R}^{\nu })_{j_{\alpha}k_{\alpha}}^{\alpha _{\nu }\alpha _{\nu }}[(a,n)]=\lambda _{\nu}(\alpha)\delta
	_{j_{\alpha}k_{\alpha}},
	\ee
	where $\alpha\in \nu$,  $j_{\alpha}$ are PRIR indices of
	the irrep $\psi _{R}^{\nu }$ and $\lambda _{\nu}(\alpha)$ is given
	Lemma~\ref{semAn}, Theorem~\ref{thm18} or equivalently in Corollary~\ref{cor_lamda}.
\end{proposition}
Proof of Proposition~\ref{propoopo} follows from Proposition~\ref{Prir1} of Appendix~\ref{AppAA} and Lemma~\ref{semAn}.

\section{Application of PRIRs to the representation theory of the algebra $\mathcal{A}_n^{t_n}(d)$}
\label{simplification}
In the following subsections we derive simpler form of the matrices $Q(\alpha)$,  $Z(\alpha)$, and projectors $F_{\nu}(\alpha)$ which where defined in paper~\cite{Moz1} (or see  Section~\ref{Sec3} of this manuscript) by use of the concept of PRIRs introduced in Section~\ref{AppA}. Additionally as a second result we present explicit and relatively simple expression for the matrix elements of the permutation operators $V'[(a,n)]$ for $a=1,\ldots,n-1$ which is now more convenient for practical use.

\subsection{Simplification of the matrices $Q(\alpha)$ and $Z(\alpha)$ and matrix representation of $V'[(a,n)]$}
\label{sim_QZ}
Let us  consider the induced representation of $S(n-1)$ $\Phi^{\alpha
}=\ind_{S(n-2)}^{S(n-1)}(\varphi^{\alpha})$, where $\varphi ^{\alpha }$ is a given irrep of $S(n-2)$. It is known that the decomposition into irreps $\psi ^{\nu }$ of $S(n-1)$  
\be
\Phi^{\alpha }=\ind_{S(n-2)}^{S(n-1)}(\varphi^{\alpha})\simeq \bigoplus _{\nu \in \Phi^{\alpha }}\psi ^{\nu }
\ee
is simple reducible. It means that there exist an unitary matrix $Z(\alpha )$, which reduces the representation to an irreducible block diagonal form. In the  paper~\cite{Moz1} such a matrix $Z(\alpha )$ was constructed, for arbitrary form of irreps $\psi ^{\nu }$ of $S(n-1)$ and irrep $\varphi ^{\alpha }$
of $S(n-2)$ (see Thm.~\ref{thm18} b) in Section~\ref{Sec3}). First let us observe  that on the $RHS$ of the main equality~\eqref{mianequality} of  Theorem~\ref{BL16} 
we get the matrix elements of induced representation $\Phi^{\alpha }=\ind_{S(n-2)}^{S(n-1)}(\varphi^{\alpha})$, i.e.  we have 
\be
\Phi^{\alpha }=\left( (\Phi^{\alpha })_{j_{\alpha }k_{\alpha }}^{ab}\right) =\left( \delta
_{a\sigma (b)}\varphi _{j_{\alpha }k_{\alpha }}^{\alpha }[(a,n-1)\sigma
(b,n-1)]\right) ,
\ee
which is a standard matrix form of the induced representation $%
\ind_{S(n-2)}^{S(n-1)}(\varphi^{\alpha} ).$
Having this we are in the position to formulate the following:
\begin{theorem}
	\label{simple_Z}
	Let $\Phi^{\alpha }=\ind_{S(n-2)}^{S(n-1)}(\varphi^{\alpha})\simeq \bigoplus _{\nu \in
		\Phi^{\alpha }}\psi _{R}^{\nu }$, where $\psi _{R}^{\nu }$ are PRIR
	representations of the group $S(n-1)$, then the corresponding matrix\ $%
	Z_{R}(\alpha )=(z_{R}(\alpha )_{kj_{\nu }}^{a\nu })$ which reduces the
	induced representation $\Phi ^{\alpha }=\ind_{S(n-2)}^{S(n-1)}\varphi
	^{\alpha }$ into the direct sum $\bigoplus _{\nu \in \Phi^{\alpha }}\psi
	_{R}^{\nu }$ has the following form:
	\be
	\label{BZR}
	z_{R}(\alpha )_{k_{\alpha }j_{\xi _{\nu }}}^{a\xi _{\nu }}=\frac{1}{\sqrt{n-1%
	}}\frac{\sqrt{d_{\nu }}}{\sqrt{d_{\alpha }}}(\psi _{R}^{\nu })_{k_{\alpha
		}j_{\xi _{\nu }}}^{\alpha \xi _{\nu }}[(a,n-1)],
	\ee
	where $(\alpha ,k_{\alpha })$ and  $(\xi _{\nu },j_{\xi _{\nu }})$  are the PRIR indices in $\psi _{R}^{\nu }$, corresponding to reducible structure for
	the subgroup $S(n-2)$ (see eq.~\eqref{eq27}). The irrep $\varphi ^{\alpha }$ of $S(n-2)$ is
	included in every $\psi _{R}^{\nu }\in \Phi ^{\alpha }$. The matrix  $%
	Z_{R}(\alpha )=\left( z_{R}(\alpha )_{k_{\alpha }j_{\xi _{\nu }}}^{a\xi _{\nu }}\right) $
	is unitary and satisfies
	\be
	\sum_{ak_{\alpha }}\sum_{bl_{\alpha }}z_R^{\dagger}(\alpha )_{j_{\xi _{\nu
		}}k_{\alpha }}^{\xi _{\nu }a}\Phi ^{\alpha }(\sigma )_{k_{\alpha }l_{\alpha
	}}^{ab}z_R(\alpha )_{l_{\alpha }j_{\zeta _{\mu }}}^{b\zeta _{\mu }}=\delta
	^{v\mu }(\psi _{R}^{\nu })_{j_{\xi _{\nu }}j_{\zeta _{\mu }}}^{\xi _{\nu
		}\zeta _{\mu }}(\sigma ),\qquad \forall \sigma \in S(n-1)
	\ee
	and 
	\be
	\sum_{ak_{\alpha }}z_R^{\dagger}(\alpha )_{j_{\xi _{\nu }}k_{\alpha }}^{\xi _{\nu
		}a}z_R(\alpha )_{k_{\alpha }j_{\zeta _{\mu }}}^{a\zeta _{\mu }}=\delta ^{v\mu
	}\delta ^{\xi _{\nu }\zeta _{\mu }}\delta _{j_{\xi _{\nu }}j_{\zeta _{\mu
	}}},
	\ee
	as well 
	\be
	\sum_{v,\xi _{\nu },j_{\xi _{\nu }}}z_R(\alpha )_{k_{\alpha }j_{\xi _{\mu
	}}}^{a\xi _{\mu }}z_R^{\dagger}(\alpha )_{j_{\xi _{\nu }}l_{\alpha }}^{\xi _{\nu
		}b}=\delta ^{a,b}\delta _{k_{\alpha }l_{\alpha }}.
	\ee
\end{theorem}

\begin{proof}
	An application of equation~\eqref{eqAP13} from Prop.~\ref{AP13} and Thm.~\ref{BL16} for PRIRs to the equation for the matrix $Z(\alpha)$ given in Thm.~\ref{thm18} b) leads directly to expressions from the statements of the theorem.
\end{proof}

One can see, that comparing expression for $Z_{R}(\alpha )$ given through equation~\eqref{BZR} with the general formula for $Z(\alpha )=(z(\alpha )_{kj_{\nu }}^{a\nu })$ in Theorem~\ref{thm18} we have substantial simplification which so important in the practical applications of our tools (see Section~\ref{Appl_PBT}).
The main advantage over the previous expression is that there is no sum over all  permutations, so we remove complexity of order $(n-1)!$. This will allow to produce expressions for matrix elements of algebra of operators  that are tractable (see Prop.~\ref{BP20}). Moreover results contained in Theorem~\ref{simple_Z} solve also the problem of the eigenvectors of the matrix $Q(\alpha)$. Namely they are given by the columns of the matrix $Z_{R}(\alpha
)=\left( z_{R}(\alpha )_{k_{\alpha }j_{\xi _{\nu }}}^{a\xi _{\nu }}\right) $, which are now relatively simple.

It is well known  that the columns of any unitary matrix form a set of
orthonormal vectors. Using this fact as a corollary from the properties of the
matrix $Z_{R}(\alpha )$ given in Thm.~\ref{simple_Z} we get the following:

\begin{corollary}
	\label{BC18}
	The set of $(n-1)d_{\alpha }$ vectors 
	\be
	T^{\nu }(\xi _{\nu },j_{\xi _{\nu }})=\left( T^{\nu }(\xi _{\nu },j_{\xi _{\nu
	}})_{k_{\alpha }}^{a}\right) =\left( z_{R}(\alpha )_{k_{\alpha }j_{\xi _{\nu }}}^{a\xi
		_{\nu }}\right) \in \mathbb{C}^{(n-1)d_{\alpha }}
	\ee
	forms an orthonormal basis of the space $\mathbb{C}^{(n-1)d_{\alpha }}.$
\end{corollary}

Next from the properties of the matrix  $Z_{R}(\alpha )$ we get the following corollary which is direct consequence of and Corollary~\ref{BC18}:

\begin{corollary}
	\label{simple1}
	The matrix  $Z_{R}(\alpha )$ diagonalises the matrix $Q(\alpha )$:
	\be
	\sum_{ak_{\alpha }}\sum_{bl_{\alpha }}z_R^{\dagger}(\alpha )_{j_{\xi _{\nu
		}}k_{\alpha }}^{\xi _{\nu }a}Q(\alpha )_{k_{\alpha }l_{\alpha
	}}^{ab}z_R(\alpha )_{l_{\alpha }j_{\zeta _{\mu }}}^{b\zeta _{\mu }}=\delta
	^{v\mu }\delta ^{\xi _{\nu }\zeta _{\mu }}\delta _{j_{\xi _{\nu }}j_{\zeta
			_{\mu }}}\lambda _{\nu}(\alpha),
	\ee
	which means that the vectors $T^{\nu }(\xi _{\nu },j_{\xi _{\nu }})$ from Corollary~\ref{BC18} are
	eigenvectors of the matrix $Q(\alpha )$. The numbers $\lambda _{\nu}(\alpha)$ are given
	Lemma~\ref{semAn}, Theorem~\ref{thm18} or equivalently in Corollary~\ref{cor_lamda}.
\end{corollary}
Next important consequence of simplification of matrix $Z(\alpha)$ by PRIR approach are relatively handy expressions for the matrix representations of the generators of algebra $\A$, especially for $V'$. Namely we have the following
\begin{proposition}
	\label{BP20}
	In the irrep $\Phi^{\alpha }$ of the algebra $%
	\mathcal{A}_{n}^{t_{n}}(d)$ we have the following matrix representation of elements $%
	V'[(a,n)]$  
	\be
	\label{blee1}
	M_{f}^{\alpha }\left[V'[(a,n)]\right]_{j_{\xi_{\omega}} \ j_{\xi_{\nu}}}^{\xi_{\omega} \ \xi_{\nu}}=\frac{1}{n-1}\frac{\sqrt{d_{\xi}d_{\omega}}}{d_{\alpha}}\sum_{k_{\alpha}}\sqrt{\lambda_{\omega}(\alpha)}\psi_{R \ j_{\xi_{\omega}} \ k_{\alpha}}^{\omega \ \xi_{\omega} \ \alpha}[(a, n-1)]\psi_{R \ k_{\alpha} \ j_{\xi_{\nu}}}^{\nu \ 
\alpha \ \xi_{\nu}}[(a, n-1)]\sqrt{\lambda_{\nu}(\alpha)},
	\ee
	where $\omega ,\nu \neq \theta$ and the subscript $f$ (see Prop.~\ref{f_basis}) means that the matrix representation is calculated
	in reduced basis $%
	f\equiv\{f_{j_{\nu} }^{\nu }:h(\nu )\leq d,\quad j_{\nu} =1,\ldots,d_{\nu }\}$ of the ideal $\Phi^{\alpha }$. 
	
	In particular for $a=n-1$ expression~\eqref{blee1} reduces to
	\be
	\label{blee2}
	M_f^{\alpha}\left(V' \right)_{j_{\xi_{\omega}} \ j_{\xi_{\nu}}}^{\xi_{\omega} \ \xi_{\nu}}=\frac{1}{n-1}\frac{\sqrt{d_{\xi}d_{\omega}}}{d_{\alpha}}\sqrt{\lambda_{\omega}(\alpha)\lambda_{\nu}(\alpha)}\delta^{\xi_{\omega}\alpha}\delta^{\xi_{\nu}\alpha}\delta_{j_{\xi_{\omega} j_{\xi_{\nu}}}}.
	\ee
\end{proposition}
Later in this paper we use simplified notation for the matrix elements $V'[(an)]$ in the reduced basis $f$:
\be
\label{notation1}
M_f^{\alpha}\left[ V'[(a,n)] \right] _{j_{\xi_{\omega}} \ j_{\xi_{\nu}}}^{\xi_{\omega} \ \xi_{\nu}}\equiv M_f^{\alpha}\left[ V'[(a,n)] \right] ^{\omega \nu}_{\xi_{\omega}j_{\xi_{\omega}} \ \xi_{\nu}j_{\xi_{\nu}}},   
\ee
where 
\begin{itemize}
	\item $\xi_{\omega}$ labels irreps of $S(n-2)$ included in $\psi^{\omega} \in \widehat{S}_d(n-1)$,
	\item $j_{\xi_{\omega}}$ labels indices in $\xi_{\omega}$,
	\item $\xi_{\nu}$ labels irreps of $S(n-2)$ included in $\psi^{\nu}\in \widehat{S}_d(n-1)$,
	\item $j_{\xi_{\nu}}$ labels indices in $\xi_{\nu}$.
\end{itemize}

Next we exploit  the idea of  PRIRs introduced in~\cite{Stu2017}  with an additional results presented in Section~\ref{AppA} to the simplification of the set of the orthogonal projections $F_{\nu}(\alpha)$ given in Prop.~\ref{complicated} onto non-trivial irreducible spaces of the algebra $\mathcal{A}_n^{t_n}(d)$. Results proven below will allow us to use them in  Section~\ref{spec_prop}, where the  properties for the PBT scheme, and its connection with the matrix $Q(\alpha)$ are delivered.  As we can see  eq.~\eqref{F_comp} is not explicit since we have to compute separately coefficients $z(\alpha)_{j_{\nu }i}^{\nu a}$ and $N^{\nu}_{j'_{\nu}}$ which are given by highly complicated equations (see Remark~\ref{Hard}). Using PRIRs  we have the following simplification:

\begin{proposition}
	\label{F_simple}
	Using PRIR representation  we can simplify form of the operators $F^{\nu }(\alpha )$  from Prop.~\ref{complicated} to explicit expression 
	 type 
	\be
	\label{FF1}
	F_{\nu }(\alpha )=\frac{1}{\lambda _{\nu}(\alpha)}\frac{d_{\nu }}{%
		(n-1)!}\sum_{b=1}^{n-1}\sum_{k_{\alpha }}\sum_{\gamma \in S(n-1)}(\psi
	_{R}^{\nu })_{k_{\alpha }k_{\alpha }}^{\alpha \alpha }(\gamma
	^{-1})V'(b,n)V[(b,n-1)\gamma (b,n-1)].
	\ee
\end{proposition}
Indeed reader notices that we simplified complicated expressions for $z(\alpha)_{j_{\nu }i}^{\nu a}$ and $N^{\nu}_{j'_{\nu}}$, and now we have  only one sum over all permutations from $S(n-1)$ instead of two of them. Next, directly from result contained in Prop.~\ref{F_simple}  we get

\begin{lemma}
	\label{lemma_Mf}
	The matrix form of the projector $F_{\nu }(\alpha )$ on non-trivial irreducible spaces of the algebra $\mathcal{A}_{n}^{t_{n}}(d)$, in the reduced basis $f$ has
	the following form
	\be
	M_{f}^{\alpha }[F_{\nu }(\alpha )]_{\xi _{\eta }j_{\xi _{\eta }}\xi _{\mu
		}j_{\xi _{\mu }}}^{\eta \qquad \mu }=\delta ^{\eta \nu }\delta ^{\nu \mu
	}\delta _{\xi _{\eta }\xi _{\mu }}\delta _{j_{\xi _{\eta }}j_{\xi _{\mu }}},
	\ee
	i.e. in the irrep $\Phi^{\alpha}$ of the algebra $\mathcal{A}_{n}^{t_{n}}(d)$, in the
	reduced basis $f$ the projector $F_{\nu }(\alpha )$ takes its canonical
	form with one$'$s on the diagonal in the position of the irrep $%
	\psi ^{\nu }$ of the group $S(n-1)$ only.
\end{lemma}

This result is obtained by a direct calculation using PRIRs. From the statement of  
Lemma~\ref{lemma_Mf}, using a decomposition of the natural representation of the
algebra $\mathcal{A}_{n}^{t_{n}}(d)$ into its irreps, we deduce easily 

\begin{corollary}
	\label{CM4}
	\be
	\tr M_{f}^{\alpha }[F_{\nu }(\alpha )]=d_{\nu },
	\ee
	and from this we get  
	\be
	\label{FF2}
	\tr_{\mathcal{H}}F_{\nu }(\alpha )=m_{\alpha }d_{\nu },
	\ee
	where $\mathcal{H}=(\mathbb{C}^d)^{\ot n}$, and $m_{\alpha}$ is the multiplicity  the irreps $\varphi ^{\alpha }$ of $S(n-2)$ in the representation $V_{d}(S(n-2))$ (see Proposition~\ref{mult_a}).
\end{corollary}

The trace $\tr_{\mathcal{H}}F_{\nu }(\alpha )$ can be computed also in another way,
directly from expression~\eqref{FF1}
\be
\tr_{\mathcal{H}}F_{\nu }(\alpha )=\frac{1}{\lambda _{\nu}(\alpha)}\frac{d_{\nu }}{(n-1)!}\sum_{b=1}^{n-1}\sum_{k_{\alpha }}\sum_{\gamma \in
	S(n-1)}(\psi _{R}^{\nu })_{k_{\alpha }k_{\alpha }}^{\alpha \alpha }(\gamma
^{-1})\tr_{\mathcal{H}}\left( V'[(b,n)]V[(b,n-1)\gamma (b,n-1)]\right) ,
\ee
where 
\be
\tr_{\mathcal{H}}\left( V'[(b,n)]V[(b,n-1)\gamma (b,n-1)]\right) =\tr_{\mathcal{H}}V[(n-1,n)\gamma].
\ee
Using the well known fact that $\tr_{\mathcal{H}}V_{n}(\sigma )=d^{l(\sigma )}=\chi
^{V_{n}(d)}(\sigma )$ is the character of the permutation representation of $S(n)$,
where $l(\sigma )$ is the number of cycles in the permutation $\sigma \in
S(n)$  we get 
\be
\tr_{\mathcal{H}}F_{\nu }(\alpha )=\frac{1}{\lambda _{\nu}(\alpha)}\frac{d_{\nu }}{(n-1)!}\sum_{b=1}^{n-1}\sum_{k_{\alpha }}\sum_{\gamma \in
	S(n-1)}(\psi _{R}^{\nu })_{k_{\alpha }k_{\alpha }}^{\alpha \alpha }(\gamma
^{-1})d^{l[(n-1,n)\gamma ]},
\ee
and it is easy to check that $l[(n-1,n)\gamma ]=l(\gamma ),\quad \gamma \in
S(n-1)$.
Further using the orthogonality relations for irreps we get one more 
	\be
	\tr_{\mathcal{H}}F_{\nu }(\alpha )=\frac{1}{\lambda _{\nu}(\alpha)}\frac{d_{\nu }}{(n-1)!}\sum_{b=1}^{n-1}\sum_{k_{\alpha }}\sum_{\gamma \in
		S(n-1)}(\psi _{R}^{\nu })_{k_{\alpha }k_{\alpha }}^{\alpha \alpha }(\gamma
	^{-1})d^{l(\gamma )}=\frac{(n-1)}{\lambda _{\nu}(\alpha)}d_{\alpha
	}m_{\nu }.
	\ee

As a corollary from  Cor.~\ref{CM4} we get the following

\begin{corollary}
	\label{cor_lamda}
The eigenvalues $\lambda_{\nu}(\alpha)$ of the matrix $Q(\alpha)$ (see Definition~\ref{def_Q}) are of the form
\be
\lambda_{\nu}(\alpha)=(n-1)\frac{m_{\nu}d_{\alpha}}{m_{\alpha}d_{\mu}}.
\ee
\end{corollary}
This an equivalent expression for the eigenvalues of $Q(\alpha)$ given in Lemma~\ref{semAn} and Thm.~\ref{thm18} a). This formula has been obtained in~\cite{Stu2017} using different method.
\subsection{New matrix operators in algebra $\A$}
\label{obiekty_sec}
In this paragraph we define new set of matrix operators which give  useful  description of the generator $V'$ of underlying algebra $\A$ in the matrix form. Derived expressions are similar to those which can be obtained for groups. Description of the latter can be find in classical textbooks~\cite{Curtis},~\cite{FHa} or in Appendix F of~\cite{Stu2017}. 
As we have shown in Section~\ref{simplification} 
in irrep of the ideal $\mathcal{M}$ labelled by partition $\alpha$  we have basis labelled by $\mu,\xi_\mu, i_{\xi_\mu}$, 
where $\mu$ is partition of $n-1$, $\xi$ is partition of $n-2$ differing from $\mu$ by one block, and $i_{\xi_\mu}$ labels basis in 
irrep $\xi_\mu$. We thus have 
a vector basis $\{|\phi^{\mu}_{\xi_{\mu} i_{\xi_{\mu}}}(\alpha,r)\>\}$, where $r$ labels multiplicity in our representation. 
With this basis we can associate flip operators (or, matrix basis operators):
\be
\label{przerzut-ogolnie}
E^{\mu \nu}_{\xi_{\mu}i_{\xi_{\mu}},\xi_{\nu}j_{\xi_{\nu}}}(\alpha)\equiv \sum_{r} |\phi^{\mu}_{\xi_{\mu} i_{\xi_{\mu}}}(\alpha,r)\>\<\phi^{\nu}_{\xi_{\nu}j_{\xi_{\nu}}}(\alpha,r)|.
\ee
In the reduced matrix basis $f$ given in Proposition~\ref{f_basis} the matrix elements of above operators are of the form
\be
	\label{mat_ele}
M_{f}^{\alpha } \left[ E^{\mu \nu}_{\xi_{\mu}i_{\xi_{\mu}},\xi_{\nu}j_{\xi_{\nu}}}(\alpha)\right]^{\rho \gamma}_{\xi_{\rho}i_{\xi_{\rho}},\xi_{\gamma}j_{\xi_{\gamma}}} =\delta^{\mu\rho}\delta^{\nu \gamma}\delta_{\xi_{\mu}\xi_{\rho}}\delta_{\xi_{\nu}\xi_{\gamma}}\delta_{i_{\xi_{\mu}}i_{\xi_{\rho}}}\delta_{j_{\xi_{\nu}}j_{\xi_{\gamma}}}.
\ee
Below we will show explicit form of some of those flip operators in terms of the elements of our algebra $\A$.

\begin{theorem}
	\label{przerzutniki}
	Let us define the following set of operators in the algebra $\A$
	\label{d1}
	\be
	\label{eq:przerzutniki_wybrane}
	E_{i_{\alpha }j_{\alpha }}^{\mu \nu }(\alpha )\equiv\frac{m_{\alpha }}{\sqrt{%
			m_{\mu }m_{\nu }}}P_{\mu }E_{i_{\alpha }j_{\alpha }}^{\alpha }V'P_{\nu }, 
	\ee
	where 
	\be
	P_{\mu }=\frac{d_{\mu }}{(n-1)!}\sum_{\sigma \in S(n-1)}\chi ^{\mu }(\sigma
	^{-1})V(\sigma ),\qquad E_{i_{\alpha }j_{\alpha }}^{\alpha }=\frac{d_{\alpha
	}}{(n-2)!}\sum_{\pi \in S(n-2)}\varphi _{j_{\alpha }i_{\alpha }}^{\alpha
	}(\pi ^{-1})V(\pi ). 
	\ee
	Then operators in~\eqref{eq:przerzutniki_wybrane} form a subset in the set of operators given through~\eqref{przerzut-ogolnie}. Clearly the operators $E_{i_{\alpha }j_{\alpha }}^{\mu \nu }(\alpha )$
	belong to the ideal $\MM$.
\end{theorem}
\begin{proof}
	To prove statement of theorem we have to show that operators given in~\eqref{eq:przerzutniki_wybrane} form a subset contained in the set composed of operators given through~\eqref{przerzut-ogolnie}. To do so we compute matrix elements of $E_{i_{\alpha }j_{\alpha }}^{\mu \nu }(\alpha )$ in reduced basis $f$ and compare them with expression~\eqref{mat_ele}. 
	In order to compute desired matrix elements first we have to calculate in PRIR
	representation  matrix elements $\psi _{R}^{\nu }$ of $S(n-1)$
	\be
	\begin{split}
	\psi _{R}^{\nu }[E_{i_{\alpha }j_{\alpha }}^{\alpha }]_{i_{\xi _{\nu
		}}j_{\zeta _{\nu }}}^{\xi _{\nu }\zeta _{\nu }}&=\delta ^{\xi _{\nu }\zeta
		_{\nu }}\varphi _{{i_{\xi _{\nu }}}j_{_{i_{\xi _{\nu }}}}}^{\xi _{\nu
	}}[E_{i_{\alpha }j_{\alpha }}^{\alpha }]=\delta ^{\xi _{\nu }\zeta _{\nu }}%
	\frac{d_{\alpha }}{(n-2)!}\sum_{\pi \in S(n-2)}\varphi _{j_{\alpha
		}i_{\alpha }}^{\alpha }(\pi ^{-1})\varphi _{{i_{\xi _{\nu }}}j_{_{\xi _{\nu
	}}}}^{\xi _{\nu }}(\pi )\\
	&=\delta ^{\xi _{\nu }\zeta _{\nu }}\delta ^{\alpha \xi _{\nu }}\delta
	_{i_{\alpha }i_{\xi _{\nu }}}\delta _{j_{\alpha }j_{\xi _{\nu }}}=\delta
	^{\alpha \zeta _{\nu }}\delta ^{\alpha \xi _{\nu }}(e_{i_{\alpha }j_{\alpha
	}})_{i_{\xi _{\nu }}j_{\xi _{\nu }}}.
\end{split}
	\ee
	Expression above means that in PRIR representation  $\psi _{R}^{\nu }$ of $S(n-1)$
	the operator $E_{i_{\alpha }j_{\alpha }}^{\alpha }$ is represented in such a
	way that among  the diagonal blocks $(\psi _{R}^{\nu })^{\xi _{\nu }\xi
		_{\nu }}\sim \varphi ^{\xi _{\nu }}$, the block $(\psi _{R}^{\nu })^{\alpha
		\alpha }\sim \varphi ^{\alpha }$ is nonzero and in this blok the operator $%
	E_{i_{\alpha }j_{\alpha }}^{\alpha }$ is represented by standard matrix
	basis $e_{i_{\alpha }j_{\alpha }}.$ From this it follows that in the irrep 
	$M_{f}^{\alpha }$ of the algebra $A_{n}^{t_{n}}(d)$ we have  
	\be
	M_{f}^{\alpha }[E_{i_{\alpha }j_{\alpha }}^{\alpha }]_{\xi _{\kappa }j_{\xi
			_{\kappa }}\xi _{\sigma }j_{\xi _{\sigma }}}^{\kappa \qquad \sigma }=\delta
	^{\kappa \sigma }\delta _{\xi _{\kappa }\alpha }\delta _{\alpha \xi _{\sigma
	}}\delta _{i_{\alpha }j_{\xi _{\kappa }}}\delta _{j_{\alpha }j_{\xi _{\sigma
	}}}=\delta ^{\kappa \sigma }\delta _{\xi _{\kappa }\alpha }\delta _{\alpha
		\xi _{\sigma }}(e_{i_{\alpha }j_{\alpha }})_{j_{\xi _{\kappa }}j_{\xi
			_{\sigma }}}.
	\ee
	Next using equation~\eqref{blee2} form Proposition~\ref{BP20} and
	\be
	M_{f}^{\alpha }[P_{\mu }]_{\xi _{\sigma }i_{\xi _{\sigma }}\xi _{\gamma
		}^{\prime }j_{\xi _{\gamma }^{\prime }}}^{\sigma \qquad \gamma }=\delta
	^{\sigma \mu }\delta ^{\mu \gamma }\delta _{\xi _{\sigma }\xi _{\gamma
		}^{\prime }}\delta _{i_{_{\xi _{\sigma }}}j_{\xi _{\gamma }^{\prime
	}}}=\delta ^{\sigma \mu }\delta ^{\mu \gamma }\delta _{\xi _{\mu }\xi _{\mu
		}^{\prime }}\delta _{i_{_{\xi _{\mu }}}j_{\xi _{\mu }}}
	\ee
	we calculate 
	\be
	\begin{split}
	&M_{f}^{\alpha }[P_{\mu }V'P_{\nu }]_{\xi _{\sigma }j_{\xi
			_{\sigma }}\xi _{\tau }j_{\xi _{\tau }}}^{\sigma \qquad \tau }=\sum_{\gamma
		\xi _{\gamma }^{\prime }j_{\xi _{\gamma }^{\prime }}}\sum_{\theta \xi
		_{\theta }j_{\xi _{\theta }}}\delta ^{\sigma \mu }\delta ^{\mu \gamma
	}\delta _{\xi _{\sigma }\xi _{\gamma }^{\prime }}\delta _{j_{_{\xi _{\sigma
		}}}j_{\xi _{\gamma }^{\prime }}}\\ 
	&\times \frac{\sqrt{m_{\gamma }m_{\theta }}}{m_{\alpha }}\delta _{\xi
		_{\gamma }^{\prime }\alpha }\delta _{\alpha \xi _{\theta }}\delta _{j_{\xi
			_{\gamma }^{\prime }}j_{\xi _{\theta }}} \delta ^{\theta \nu }\delta
	^{\nu \tau }\delta _{\xi _{\theta }\xi _{\tau }}\delta _{j_{_{\xi \theta
		}}j_{\xi _{\tau }}}=\frac{\sqrt{m_{\mu }m_{\nu }}}{m_{\alpha }}\delta
	^{\sigma \mu }\delta ^{\nu \tau }\delta _{\xi _{\sigma }\alpha }\delta
	_{\alpha \xi _{\tau }}\delta _{j_{\xi _{\sigma }}j_{\xi _{\tau }}}.
	\end{split}
	\ee
	Taking all together we write
	\be
	M_{f}^{\alpha }[P_{\mu }E_{i_{\alpha }j_{\alpha }}^{\alpha }V'P_{\nu }]_{\xi _{\rho }j_{\xi _{\rho }}\xi _{\gamma }j_{\xi _{\gamma
	}}}^{\rho \qquad \gamma }=\frac{\sqrt{m_{\mu }m_{\nu }}}{m_{\alpha }}\delta
	^{\rho \mu }\delta ^{\nu \gamma }\delta _{\xi _{\rho }\alpha }\delta
	_{\alpha \xi _{\gamma }}(e_{i_{\alpha }j_{\alpha }})_{j_{_{\xi _{\rho
		}}}j_{\xi _{\gamma }}}.
	\ee
	Finally matrix elements of the operators given in~\eqref{eq:przerzutniki_wybrane} are of the form
	\be
	\label{p1}
	\begin{split}
	M_{f}^{\alpha }[E_{i_{\alpha }j_{\alpha }}^{\mu \nu }(\alpha )]_{\xi _{\rho
		}j_{\xi _{\rho }},\xi _{\gamma }j_{\xi _{\gamma }}}^{\rho \gamma }&=\frac{\sqrt{m_{\mu }m_{\nu }}}{m_{\alpha }}\delta
	^{\rho \mu }\delta ^{\mu \gamma }\delta _{\xi _{\rho }\alpha }\delta
	_{\alpha \xi _{\gamma }}\delta _{i_{\alpha }j_{_{\xi _{\rho }}}}\delta
	_{j_{\alpha }j_{\xi _{\gamma }}}\\
	&=\frac{\sqrt{m_{\mu }m_{\nu }}}{m_{\alpha }}\delta ^{\rho \mu }\delta ^{\mu \gamma
	}\delta _{\xi _{\rho }\alpha }\delta _{\alpha \xi _{\gamma }}(e_{i_{\alpha
		}j_{\alpha }})_{j_{_{\xi _{\rho }}}j_{\xi _{\gamma }}}.
	\end{split}
	\ee
	We see that in the matrix $M_{f}^{\alpha }[E_{i_{\alpha }j_{\alpha
	}}^{\mu \nu }(\alpha )]$ the only non-zero block is the block with indices $%
	(\mu ,\nu )$ and inside this block the only non-zero subblock has indices $%
	(\alpha ,\alpha )$ and this subblock is equal to the standard matrix basis
	element $e_{i_{\alpha }j_{\alpha }}$. Comparing equation~\eqref{p1} with expression for matrix elements in~\eqref{mat_ele} we see that indeed operators from~\eqref{eq:przerzutniki_wybrane} form a subset in the set of operators given  in~\eqref{przerzut-ogolnie}. This finishes the proof.
\end{proof}
If we define in standard way 
\begin{definition}
	\be
	P_{\alpha }\equiv \sum_{i_{\alpha }}E_{i_{\alpha }i_{\alpha }}^{\alpha }=\frac{%
		d_{\alpha }}{(n-2)!}\sum_{\pi \in S(n-2)}\chi ^{\alpha }(\pi ^{-1})V(\pi ).
	\ee
\end{definition}
Then we get the following corollary from equation~\eqref{p1}
\begin{corollary}
	\be
	M_{f}^{\alpha }(P_{\mu }P_{\alpha }V'P_{\nu })=M_{f}^{\alpha
	}(P_{\mu })M_{f}^{\alpha }(V')M_{f}^{\alpha }(P_{\nu }),
	\ee
	where 
	\be
	M_{f}^{\alpha }(P_{\mu })=(\delta ^{\mu \rho }\delta ^{\mu \gamma }\delta
	_{\xi _{\rho }\xi _{\gamma }}\delta _{j_{_{\xi _{\rho }}}j_{\xi _{\gamma
	}}}),
	\ee
	so it is a block diagonal matrix with only one non-zero diagonal block with
	indices $(\mu,\mu )$. The operator $P_{\alpha }$ vanishes on $RHS$
	because, from structure of $M_{f}^{\alpha }(V')$ we have 
	\be
	M_{f}^{\alpha }(P_{\alpha })M_{f}^{\alpha }(V')=M_{f}^{\alpha
	}(V').
	\ee
\end{corollary}
Directly from Theorem~\ref{przerzutniki} and properties of $P_{\mu }$, and $E_{i_{\alpha
	}j_{\alpha }}^{\alpha }$ it follows
\begin{remark}
	\label{RR39}
	The operator $E_{i_{\alpha }j_{\alpha }}^{\mu \nu }(\alpha )$ is non-zero if and only if
	the irreps labelled by partitions $\alpha $ and $\mu ,\nu $ are in the relation 
	\be
	\mu =\alpha +\square \quad \wedge \quad \nu =\alpha +\square.
	\ee
	The total number of non-zero operators, for a given $\alpha ,$ is equal to 
	$(d_{\alpha }N[M_{f}^{\alpha }:S(n-1)])^{2}$, where $N[M_{f}^{\alpha
	}:S(n-1)] $ is the number of irreps of $S(n-1)$ in the irrep $%
	M_{f}^{\alpha }$ of the algebra $\A$.
\end{remark}
By Theorem~\ref{przerzutniki}, the operators $E_{i_{\alpha }j_{\alpha }}^{\mu \nu }(\alpha )$ satisfy the
following multiplication rule
\be
\label{mult_E}
	E_{i_{\alpha }j_{\alpha }}^{\mu \nu }(\alpha )E_{k_{\beta }l_{\beta }}^{\xi
		\omega }(\beta )=\delta _{\alpha \beta }\delta ^{\nu \xi }\delta
	_{j_{\alpha }k_{\beta }}E_{i_{\alpha }l_{\beta }}^{\mu \omega }(\alpha ),
\ee
We shall now show it also directly, by using their expression in terms of algebra elements 
\eqref{eq:przerzutniki_wybrane}.
	We have
	\be
	\begin{split}
	E_{i_{\alpha }j_{\alpha }}^{\mu \nu }(\alpha )E_{_{k\beta }l_{\beta }}^{\rho
		\sigma }(\beta )&=\frac{m_{\alpha }}{\sqrt{m_{\mu }m_{\nu }}}\frac{m_{\beta }}{\sqrt{m_{\rho
			}m_{\sigma }}}P_{\mu }E_{i_{\alpha }j_{\alpha }}^{\alpha }V'\delta _{\nu \rho }P_{\rho }E_{k_{\beta }l_{\beta }}^{\beta
	}V'P_{\sigma }\\
	&=\delta _{\nu \rho }\delta ^{\alpha \beta }\delta _{j_{\alpha }k_{\beta }}%
	\frac{m_{\alpha }}{\sqrt{m_{\mu }m_{\nu }}}\frac{m_{\beta }}{\sqrt{m_{\rho
			}m_{\sigma }}}P_{\mu }E_{i_{\alpha }l_{\alpha }}^{\alpha }V'P_{\rho }V'P_{\sigma }.
\end{split}
	\ee
	To obtain desired result the most important is to calculate $V'P_{\rho }V'$:
	\be
	\label{char}
	\begin{split}
	V'P_{\rho }V'&=\frac{d_{\rho }}{(n-1)!}%
	\sum_{\sigma \in S(n-1)}\chi ^{\rho }(\sigma ^{-1})V'V(\sigma
	)V'\\
	=&\frac{d_{\rho }}{(n-1)!}\sum_{a=1}^{n-1}\sum_{\pi \in S(n-2)}\chi ^{\rho
	}[(a,n-1)\pi ]V'V[\pi ^{-1}(a,n-1)]V'.
\end{split}
	\ee
	Now using the PRIR structure we expand the character on RHS of~\eqref{char} 
	\be
	\chi ^{\rho }[(a,n-1)\pi ]=\sum_{\xi _{\rho }j_{\xi _{\rho }}}\sum_{\gamma
		_{\rho }k_{\gamma _{\rho }}}(\psi _{R}^{\rho })_{j_{\xi _{\rho }}k_{\gamma
			_{\rho }}}^{\xi _{\rho }\gamma _{\rho }}[(a,n-1)]\delta ^{\xi _{\rho }\gamma
		_{\rho }}\varphi _{k_{\gamma _{\rho }}j_{\xi _{\rho }}}^{\gamma _{\rho
	}}(\pi )
	\ee
	and together with the identity $V'V[(a,n-1)]V'=d^{\delta _{a,n-1}}V'$
	we obtain
	\be
	\begin{split}
	&V'P_{\rho }V'=\\
	&=\frac{d_{\rho }}{(n-1)!}\sum_{a=1}^{n-1} \sum_{\pi \in S(n-2)}\sum_{\xi
		_{\rho }j_{\xi _{\rho }}}\sum_{\gamma _{\rho }k_{\gamma _{\rho }}}(\psi
	_{R}^{\rho })_{j_{\xi _{\rho }}k_{\gamma _{\rho }}}^{\xi _{\rho }\gamma
		_{\rho }}[(a,n-1)]\delta ^{\xi _{\rho }\gamma _{\rho }}\varphi _{k_{\gamma
			_{\rho }}j_{\xi _{\rho }}}^{\gamma _{\rho }}(\pi )d^{\delta
		_{a,n-1}}V'V(\pi ^{-1}).
	\end{split}
	\ee
	Further using  Proposition~\ref{propoopo} and the definition of $P_{\xi _{\rho }}$
	we find 
	\be
	V'P_{\rho }V'=\sum_{\xi _{\rho }\in \rho }%
	\frac{m_{\rho }}{m_{\xi _{\rho }}}P_{\xi _{\rho }}V',
	\ee
	Finally all together  implies that 
	\be
	E_{i_{\alpha }j_{\alpha }}^{\mu \nu }(\alpha )E_{_{k\beta }l_{\beta }}^{\rho
		\sigma }(\beta )=\delta _{\nu \rho }\delta ^{\alpha \beta }\delta
	_{j_{\alpha }k_{\beta }}\frac{m_{\alpha }}{\sqrt{m_{\mu }m_{\sigma }}}P_{\mu
	}E_{i_{\alpha }l_{\alpha }}^{\alpha }V'P_{\sigma }=\delta
	_{\nu \rho }\delta ^{\alpha \beta }\delta _{j_{\alpha }k_{\beta
	}}E_{_{i_{\alpha }}l_{\alpha }}^{\mu \sigma }(\alpha ).
	\ee
Of course since operators in~\eqref{eq:przerzutniki_wybrane} form a subset in the set of operators given by~\eqref{przerzut-ogolnie} we can also prove composition rule~\eqref{mult_E} using directly definition from~\eqref{przerzut-ogolnie}.
Using properties of $P_{\mu },P_{\alpha}$, and $E_{i_{\alpha }j_{\alpha }}^{\alpha }$ one
can deduce decomposition of $V'$ in terms of matrix operators given in Definition~\ref{d1}, namely we have
\begin{proposition}
	Operator $V'$ have the following decomposition in terms of operators $E_{i_{\alpha }i_{\alpha }}^{\alpha }$ given in Theorem~\ref{d1}
	\be
	V'=\sum_{\alpha }\sum_{\mu, \nu \in \alpha }\sum_{i_{\alpha }}\frac{\sqrt{%
			m_{\mu }m_{\nu }}}{m_{\alpha }}E_{i_{\alpha }i_{\alpha }}^{\mu \nu }(\alpha
	),
	\ee
	where summations are over partitions for which $\varphi^{\alpha} \in \widehat{S}_d(n-2)$ and 
	$\psi^{\mu} ,\psi^{\nu} \in \widehat{S}_d(n-1)$.
\end{proposition}
Now if we define new set of operators as
\be
\label{P}
\begin{split}
	P^{\mu \nu}(\alpha)\equiv \sum_{i_{\alpha}}E_{i_{\alpha}i_{\alpha}}^{\mu \nu}(\alpha)=\frac{m_{\alpha}}{\sqrt{m_{\mu}m_{\nu}}}P_{\mu}\left(\sum_{i_{\alpha}}E_{i_{\alpha}i_{\alpha}}^{\alpha} \right)V'P_{\nu}
	=\frac{m_{\alpha}}{\sqrt{m_{\mu}m_{\nu}}}P_{\mu}P_{\alpha}V'P_{\nu},
\end{split}
\ee
then together with expression~\eqref{mult_E} we can formulate the following
\begin{corollary}
	The operators $P^{\mu \nu}(\alpha)$ given in~\eqref{P} satisfy
	\be
	P^{\mu \nu}(\alpha)P^{\xi \omega}(\beta)=\delta^{\nu \xi}\delta_{\alpha \beta}P^{\mu \omega}(\beta).
	\ee
\end{corollary}

The operators $E_{i_{\alpha }j_{\alpha }}^{\mu \nu }(\alpha )$ for a
given $\alpha ,$ form an algebra isomorphic with the matrix algebra $%
\M(d_{\alpha }N[M_{f}^{\alpha }:S(n-1)],\mathbb{C})\subset \M(d^{n},\mathbb{C})$. 
Directly from the multiplication rule for the operators $E_{i_{\alpha
	}j_{\alpha }}^{\mu \nu }(\alpha )$ we have also

\begin{corollary}
	The subset of $d_{\alpha}^{2}$ operators of the form $E_{i_{\alpha
		}j_{\alpha }}^{\nu \nu }(\alpha )$ satisfy 
	\be
	E_{i_{\alpha }j_{\alpha }}^{\nu \nu }(\alpha )E_{_{k_{\alpha }}l_{\alpha
	}}^{\nu \nu }(\alpha )=\delta _{j_{\alpha }k_{\alpha }}E_{i_{\alpha
		}l_{\alpha }}^{\nu \nu }(\alpha ), 
	\ee
	so it forms a subalgebra isomorphic with the matrix algebra $\M(d_{\alpha },\mathbb{C})$, but again as matrices the operators $E_{i_{\alpha }j_{\alpha }}^{\nu \nu
	}(\alpha )$ belong to matrix algebra $\M(d^{n},\mathbb{C})$.
\end{corollary}
From equation~\eqref{mult_E} it follows that operators $E_{i_{\alpha }j_{\alpha }}^{\mu
	\nu }(\alpha )$ may be represented by the standard elementary matrices $%
e_{ij}\in \M(d_{\alpha }N[M_{f}^{\alpha }:S(n-1)],\mathbb{C})$, 
but this matrix representation is not equivalent to matrix representation
in the irreps $M_{f}^{\alpha }$ of the algebra $\A$.

\section{Application to deterministic port-based teleportation}
\label{Appl_PBT}
\bigskip As we mentioned at the beginning algebraic tools described in previous sections have explicit connection with a novel  port-based teleportation protocols (PBT) introduced and analysed for qubit case in the series of papers~\cite{ishizaka_asymptotic_2008,ishizaka_quantum_2009,ishizaka_remarks_2015} and extended to the qudit case: partially using graphical representation of Temperley-Lieb algebra in~\cite{wang_higher-dimensional_2016} and fully in~\cite{Stu2017,Moz2017b}. It can be shown that description  of the probabilistic and deterministic PBT can be write down purely in terms of characteristic of the algebra $\A$ and it is not only different description of the problem. Namely only using representation theory of $\A$ we are able to describe PBT in dimensions higher than two in a efficient way. 

\subsection{Port-based Teleportation Protocol and Algebra $\mathcal{A}_n^{t_n}(d)$}
Hereunder  we give a brief description of the deterministic port-based teleportation (dPBT) protocol and its connection with the algebra of partially transposed permutation operators  with respect to last subsystem. 
We present connection between PBT operator encoding the performance of teleportation protocol and matrices $Q(\alpha)$ which encode properties of underlying algebra $\A$. Further we apply simplified formalism of the algebra $\A$ to present alternative proofs of theorems for entanglement fidelity $F$ in the case of dPBT.  Using PRIR basis we solve an eigen-problem for the generators $V'[(a,n)]$, where $a=1,\ldots,n-1$ which allows us to present discussion about asymptotic behaviour of $F$ as well re-derive lower bound on $F$ founded previously in~\cite{beigi_konig}. It is worth to mention that having full spectral analysis of $V'[(a,n)]$ gives us possibilities for further investigations of hybrid scheme of PBT for qudits~\cite{ishizaka_remarks_2015}. 

In the standard version of the PBT protocol Alice and Bob share a large resource state composed of $N$ copies of the maximally entangled state $|\psi^+\>$. Each copy is a two-qudit state named as port. Alice wishes to teleport to Bob an unknown state $\theta$. Do do so she performs a joint measurement from the set of POVM $\{\Pi_i\}_{i=1}^N$ on her half of the resource state and the unknown state $\theta$ and obtains an outcome $i\in\{1,\ldots,N\}$, and communicates it to Bob. When Bob receives the information from Alice he discards all the ports except $i$-th which is the teleported state. The most important features of this kind of protocol is lack of unitary correction as in the ordinary teleportation scheme~\cite{bennett_teleporting_1993} and fact that teleported state is always successfully teleported but it is distorted. It means that the fidelity $F$ between unknown state $\theta$ and teleported state is smaller than one and is function of $N$ as well as local Hilbert space dimension $d$.
Namely we have the following expression for the fidelity in the mentioned scenario
\be
\label{F_1}
F=\frac{1}{d^2}\sum_{i=1}^N\tr\left[\sigma_i \widetilde{\rho}^{-1/2}\sigma_i \widetilde{\rho}^{-1/2} \right], 
\ee 
where $N=n+1$.
In the above $\widetilde{\rho}$ is called PBT operator and together with operators $\sigma_i$ has the following representation in terms of partially transposed permutation operators $V'[(i,n)]$ for $i=1,\ldots,n-1$
\be
\label{PBT_op}
\sigma_i=\frac{1}{d^N}V'[(i,n)],\quad \widetilde{\rho}=\sum_{i=1}^N\sigma_i=\frac{1}{d^N}\sum_{i=1}^NV'[(i,n)].
\ee 
From this construction it follows clearly that PBT operator $\rho$ is an element of previously studied algebra of the partially transposed permutation operators with respect to last subsystem $\A$.   We know that algebra $\A$ decomposes into direct sum of two ideals, i.e. $\A=\mathcal{M}\oplus \mathcal{S}$ and for analysis of the dPBT scheme knowledge only about ideal $\mathcal{M}$ is crucial. What is the most important (see~\cite{Stu2017}) that due to symmetries in the system operator $\rho$ is diagonal in the blocks represented by projectors $F_{\nu}(\alpha)$ spanning irreps  contained in $\mathcal{M}$ of the algebra $\A$ and can be written as
\be
\label{pbt1}
\widetilde{\rho}=\frac{1}{d^N}\rho=\frac{1}{d^N}\sum_{\alpha:h(\alpha)\leq d} \ \sum_{\mu \in \Phi^{\alpha}}\lambda_{\nu}(\alpha)F_{\nu}(\alpha),
\ee
where the numbers
\be
\lambda_{\nu}(\alpha)=N\frac{m_{\nu}d_{\alpha}}{m_{\alpha}d_{\nu}}
\ee
are eigenvalues of the sum $\rho=\sum_{i=1}^NV'[(i,n)]$.
\subsection{Spectrum of the PBT operator and its connection with matrix $Q(\alpha)$}
\label{spec_prop}
In this section we focus on spectral analysis of the operator $\rho$ given in equation~\eqref{pbt1}.
The first spectral property of the operator $\rho$, which one
can easily calculate using Prop.~\ref{mult_a} is the following

\begin{proposition}
\be
\tr \rho =(n-1)d^{n-1}, 
\ee
where the trace is taken in the space $(\mathbb{C}^{d})^{\otimes n}.$
\end{proposition}

The next step to describe the spectrum of $\rho$ is to find the
matrix form $M^{R}(\rho)$, where $R\in\{\Phi ^{\alpha },\Psi ^{\nu }\}$, of this the operator $\rho$ in the irreps $\Phi ^{\alpha },\Psi ^{\nu }$ of the
algebra $\mathcal{A}_{n}^{t_{n}}(d)$. This of course depends on choice of the basis in the irreps $\Phi ^{\alpha },\Psi ^{\nu }$. For example in the irrep $\Phi^{\alpha }$ in the basis $E(\alpha )=\{e_{ij}^{ab}(\alpha ):a,b=1,\ldots,n-1, \ i,j=1,\ldots,d_{\alpha }\}$ (see~\cite{Moz1}) we have 
\be
M_{E(\alpha )}^{\Phi ^{\alpha }}(\rho)=Q(\alpha ).
\ee
Above result holds only when $d\geq n-1,$ nevertheless it shows a
connection between spectra of matrices $Q(\alpha )$ and $\rho$.
Using so called reduced basis $%
f\equiv\{f_{j_{\nu} }^{\nu }:h(\nu )\leq d,\quad j_{\nu} =1,\ldots,d_{\nu }\}$ in the irreps $\Phi ^{\alpha }$
and arbitrary basis in the irreps $\Psi ^{\nu }$ one can prove much
stronger result

\begin{proposition}
For any $n,d\geq 2$ we have 
\be
M_{f}^{\Phi ^{\alpha }}(\rho )=\operatorname{diag}(\lambda _{\nu
}(\alpha ))\in \M(\rank Q(\alpha ),\mathbb{C}), 
\ee
where $\lambda _{\nu }(\alpha )$ are all non-zero eigenvalues of the matrix $%
Q(\alpha )$ including their multiplicities. For the irreps of second kind 
$\Psi ^{\nu }$\ we have 
\be
M_{B}^{\Psi ^{\nu }}(\rho )=0\in \M(d_{\nu },\mathbb{C}), 
\ee
for any basis $B$ in the irrep $\Psi ^{\nu }$. 
\end{proposition}

From the above proposition and from Theorem~\ref{decompA}  we
deduce one of the main result of this manuscript, the structure of the spectrum of the
operator $\rho$.

\begin{theorem}
Let $\mathcal{A}_{n}^{t_{n}}(d)$ be the algebra of partially transposed operators in $(
\mathbb{C}^{d})^{\otimes n}$. Then for any $n,d\geq 2$ the non-zero eigenvalues of the
operator $\rho$ are the nonzero eigenvalues $\lambda _{\nu
}(\alpha ),$ including multiplicity, of the matrices $Q(\alpha )$, where irreps $\alpha $ of $S(n-2)$ are those which appear in the decomposition of $
\mathcal{A}_{n}^{t_{n}}(d)$ in Theorem~\ref{decompA} i.e. when $h(\alpha )\leq d$. The multiplicity $%
m_{\nu ,\alpha }$ of the eigenvalue $\lambda _{\nu }(\alpha )$ of operator $
\rho$ is equal to
\be
m_{\nu ,\alpha }=m_{\alpha }d_{\nu }, 
\ee
where $m_{\alpha}$ is given in Theorem~\ref{decompA} and the multiplicity $m_{0}$ of the
eigenvalue $0$ of operator $\rho$ is equal to%
\be
m_{0}=d^{n-1}(d-n+1). 
\ee
\end{theorem}
The explicit equation for the eigenvalues $\lambda _{\nu }(\alpha )$ of the
operator $\rho$ are given in the Lemma~\ref{semAn}, Theorem~\ref{thm18} b), and Corollary~\ref{cor_lamda}.
From the properties of the eigenvalues of the matrices $Q(\alpha )$ one can deduce several properties of the spectrum of the operator $\rho$. The first
such a result concerns the spectral radius of the operator $\rho$.

\begin{proposition}
For any $n,d\geq 2$ the biggest eigenvalue of the operator $\rho$ is of the form
\be
\lambda _{\max }=d+n-2, 
\ee
and has multiplicity $m_{\max }=\frac{1}{(n-2)!}\sum_{\sigma \in
S(n-2)}d^{l(\sigma )}.$
\end{proposition}
Above proposition easily follows from the form of $\lambda_{\nu}(\alpha)$ given in Thm.~\ref{thm18} a).
We can also calculate the minimal nonzero eigenvalue of the operator $
\rho$.

\begin{proposition}
The minimal nonzero eigenvalue of of the operator $\rho$ has the
following value
\be
\lambda _{\min }=d-h(\alpha ):h(\alpha )<d.
\ee
In particular we have:
\begin{enumerate}[a)]
\item if $d\geq n-1,$ then $\lambda _{\min }=d-h(\alpha )\geq 1$ with the multiplicity equal to 
\be
\label{mmm}
m_{\min }=\frac{1}{(n-2)!}\sum_{\sigma \in S(n-2)}\operatorname{sgn}(\sigma )d^{l(\sigma
)}, 
\ee
\item if $d\leq n-2,$ then $\lambda _{\min }=d-h(\alpha )=1$ with the multiplicity given thorough expression~\eqref{mmm}.
\end{enumerate}
\end{proposition}

\subsection{Fidelity calculation in case of maximally entangled state as a resource state}
\label{det_PBT}
Having all tools developed in the previous chapters we are ready to apply them to description of deterministic port-based teleportation.
In the first step using Proposition~\ref{BP20} we calculate the following quantity which appears in  the expression for  the fidelity $F$ given in equation~\eqref{F_1} of the deterministic version of the protocol:
\be
\tr\left[ M_{f}^{\alpha }(\rho ^{-1/2})M_{f}^{\alpha
}(V'[(a,n)])M_{f}^{\alpha }(\rho ^{-1/2})M_{f}^{\alpha
}(V'[(a,n)])\right] \equiv \tr_{\Phi _{f}^{\alpha }}\left[\rho ^{-1/2
}V'[(a,n)]\rho ^{-1/2}V'[(a,n)]\right].
\ee
Namely we have the following:
\be
\begin{split}
	&\tr_{\Phi _{f}^{\alpha }}\left[\rho ^{-1/2
	}V'[(a,n)]\rho^{-1/2}V'[(a,n)]\right] =\\
	&=\sum_{\nu ,\xi _{\nu },j_{\xi _{\nu }}} \ \sum_{_{\mu },\zeta _{\mu },j_{\zeta
			_{\mu }}} \ \sum_{k_{\alpha },l_{\alpha }}\frac{1}{(n-1)^{2}}\frac{d_{\nu
		}d_{\mu }}{d_{\alpha }^{2}}(\psi _{R}^{\nu })_{l_{\alpha }j_{\xi _{\nu
	}}}^{\alpha \xi _{\nu }}[(a,n-1)]\sqrt{\lambda _{\nu}(\alpha)}(\psi
	_{R}^{\nu })_{j_{\xi _{\nu }}k_{\alpha }}^{\xi _{\nu }\alpha }[(a,n-1)]\times \\
	&\times (\psi _{R}^{\nu })_{k_{\alpha }j_{\zeta _{\mu }}}^{\alpha \zeta _{\mu
	}}[(a,n-1)]\sqrt{\lambda _{\mu}(\alpha)}(\psi _{R}^{\nu })_{j_{\zeta
			_{\mu }}k_{\alpha }}^{\zeta _{\mu }\alpha }[(a,n-1)],
\end{split}
\ee
which, after summation over $\xi _{\nu },j_{\xi _{\nu }}$ and $\zeta _{\mu
},j_{\zeta _{\mu }}$  reduces  to 
\be
\sum_{\nu ,\mu } \ \sum_{k_{\alpha },l_{\alpha }}\frac{1}{(n-1)^{2}}\frac{
	d_{\nu }d_{\mu }}{d_{\alpha }^{2}}\sqrt{\lambda _{\nu}(\alpha)\lambda _{\mu}(\alpha)}\delta _{k_{\alpha }l_{\alpha }}\delta
_{k_{\alpha },l_{\alpha }}=\frac{1}{(n-1)^{2}}\sum_{\nu ,\mu }\frac{d_{\nu
	}d_{\mu }}{d_{\alpha }}\sqrt{\lambda _{\nu}(\alpha)\lambda _{\mu}(\alpha)}.
\ee
We can summarize above calculations in the following proposition:

\begin{proposition}
	In the irrep $\Phi^{\alpha }$ of the algebra $%
	\mathcal{A}_{n}^{t_{n}}(d)$ we have 
	\be
	\tr_{\Phi^{\alpha }}\left[\rho ^{-1/2}V'[(a,n)]
	\rho ^{-1/2}V'[(a,n)]\right] =\frac{1}{(n-1)^{2}}\sum_{\nu ,\mu }\frac{
		d_{\nu }d_{\mu }}{d_{\alpha }}\sqrt{\lambda _{\nu}(\alpha)\lambda _{\mu}(\alpha)}.
	\ee
\end{proposition}

In order to calculate the trace of the operator $\rho ^{-1/2}V'[(a,n)]\rho ^{-1/2}V'[(a,n)]$ in
arbitrary representation $R(\mathcal{A}_{n}^{t_{n}}(d))$ it is enough to multiply $
\tr_{\Phi^{\alpha }}\left[\rho ^{-1/2}V'[(a,n)]
\rho ^{-1/2}V'[(a,n)]\right] $ by the multiplicity of the irrep $\Phi^{\alpha }$ in the representation $R(\mathcal{A}_{n}^{t_{n}}(d))$, i.e. we have 

\begin{corollary}
	\label{Acor48}
	For any representation $R(\mathcal{A}_{n}^{t_{n}}(d))$ of the algebra $\mathcal{A}_{n}^{t_{n}}(d)
	$ we have 
	\be
	\tr_{R(\mathcal{A}_{n}^{t_{n}}(d))}\left[\rho ^{-1/2
	}V'[(a,n)]\rho ^{-1/2}V'[(a,n)]\right] =\frac{1}{
		(n-1)^{2}}\sum_{\Phi^{\alpha }\in R(\mathcal{A}_{n}^{t_{n}}(d))} \ \sum_{\nu ,\mu \in
		\Phi^{\alpha }}\frac{d_{\nu }d_{\mu }}{d_{\alpha }}\sqrt{\lambda _{\nu}(\alpha)\lambda _{\mu}(\alpha)}.
	\ee
	In particular we have:
	
	\begin{enumerate}[(a)]
		\item for regular representation of the algebra $\mathcal{A}_{n}^{t_{n}}(d)$ 
		\be
		\tr_{\mathcal{A}_{n}^{t_{n}}(d)}\left[\rho ^{-1/2}V'[(a,n)]
		\rho ^{-1/2}V'[(a,n)]\right] =\frac{1}{n-1}\sum_{\alpha:h(\alpha )\leq
			d} \ \sum_{\nu ,\mu \in \Phi^{\alpha }}d_{\nu }d_{\mu }\sqrt{\lambda _{\nu}(\alpha)\lambda _{\mu}(\alpha)},
		\ee
		
		\item for the natural representation of the algebra $\mathcal{A}_{n}^{t_{n}}(d)$ in the
		space $\mathcal{H}=(\mathbb{C}^{d})^{\otimes n}$
		\be
		\label{Acor48b}
		\tr_{\mathcal{H}}\left[\rho ^{-1/2}V'[(a,n)]
		\rho ^{-1/2}V'[(a,n)]\right] =\frac{1}{(n-1)^{2}}\sum_{\alpha:h(\alpha )\leq
			d} \ \sum_{\nu ,\mu \in \Phi^{\alpha }}\frac{m_{\alpha }}{d_{\alpha }}%
		d_{\nu }d_{\mu }\sqrt{\lambda _{\nu}(\alpha)\lambda _{\mu}(\alpha)},
		\ee
		where 
		\be
		m_{\alpha }=\frac{1}{(n-2)!}\sum_{\sigma \in S(n-2)}\chi ^{\alpha }(\sigma
		^{-1})d^{l(\sigma )}.
		\ee
	\end{enumerate} 
Equation~\eqref{Acor48b} leads to the following expression for the fidelity in the deterministic PBT scheme when the resource state is a maximally entangled state, which was obtained independently in~\cite{Stu2017}:
\be
F=\frac{1}{d^{N+2}}\sum_{\alpha:h(\alpha)\leq d}\left(\sum_{\mu \in \Phi^{\alpha}}\sqrt{d_{\mu}m_{\mu}} \right)^2 ,
\ee
where sums over $\alpha$ and $\mu$ are taken, whenever number of rows in corresponding Young diagrams is not greater than  the dimension of the local Hilbert space $d$. 
\end{corollary}
Reader notices that formula for entanglement fidelity presented above is for the PBT operator $\widetilde{\rho}$, so the eigenvalues $\lambda_{\nu}(\alpha)$ used for calculations have to be rescaled by the factor $1/d^N$.
\subsection{Properties of the fidelity.}

In this section we derive some basic properties of the standard fidelity given as
\be
F\equiv F_{n}(d)=\frac{n-1}{d^{n+1}}\tr_{\mathcal{H}}\left[ \rho ^{-\frac{1}{2}}V'[(a,n)]\rho ^{-\frac{1}{2}
}V'[(a,n)]\right] =\frac{n-1}{d^{n+1}}\sum_{\alpha :h(\alpha )\leq d}\frac{%
	d_{\alpha }}{m_{\alpha }}\left[ \sum_{\mu \neq \theta }\lambda _{\mu}^{-
	\frac{1}{2}}(\alpha )m_{\mu }\right] ^{2}.
\ee
It appears that the power $-\frac{1}{2}$ of $\rho $ implies very particular
properties of the fidelity $F_{n}(d)$. First of all we prove the following
\begin{theorem}
	\label{limits}
	For any $d\geq 2$ and $n\geq 2$ we have
	\be
F_{n}(d)\leq 1,\qquad \lim_{n\rightarrow \infty }F_{n}(d)=1.
	\ee
\end{theorem}

\begin{remark}
	The first property of the fidelity $F_{n}(d)$ is in
	fact, the justification of the definition of the fidelity but it also a
	necessary statement in the proof of the second result in Theorem~\ref{limits}.
\end{remark}

In the proof of Theorem~\ref{limits} we  need the spectral decomposition the
essential projectors $V'[(a,n)]$ in the irrrep $M_{f}^{\alpha }$.

\begin{proposition}
	\label{eigen1}
	The set of orthonormal vectors 
	\be
	w_{i_{\alpha }}=w_{i_{\alpha }}[a,\alpha ]=\left( \sqrt{\frac{m_{\rho }}{
			dm_{\alpha }}}\psi _{Rj_{\xi _{\rho }}i_{\alpha }}^{\rho \xi _{\rho }\alpha
	}[(a,n-1)]\right) \in \mathbb{C}^{\dim M_{f}^{\alpha }},\quad i_{\alpha }=1,\ldots,d_{\alpha }, 
	\ee
	where $\rho =\alpha +\square $, $\rho \neq \theta $, $\xi _{\rho }=\rho
	-\square ,$ $j_{\xi _{\rho }}=1,\ldots,\dim \xi _{\rho }$ are $PRIR$ indices,
	are eigenvectors of the matrix $M_{f}^{\alpha }(V'[(a,n)])$, i.e. we have
	\be
	M_{f}^{\alpha }(V'[(a,n)])w_{i_{\alpha }}[a,\alpha ]=dw_{i_{\alpha
	}}[a,\alpha ], 
	\ee
	and 
	\be
	M_{f}^{\alpha }(V'[(a,n)])=d\sum_{i_{\alpha }=1}^{d_{\alpha
	}}w_{i_{\alpha }}[a,\alpha ]w_{i_{\alpha }}^{\dagger}[a,\alpha ]. 
	\ee
	The remaining orthonormal eigenvectors of the matrix $M_{f}^{\alpha }(V'[(a,n)])$, corresponding to the eigenvalue $0$ will be denoted as $
	w_{j}[a,\alpha ],$ where $j=d_{\alpha }+1,\ldots,\dim M_{f}^{\alpha }$.
\end{proposition}

We define also

\begin{definition}
	\label{def5}
   The rectangular matrix 
	\be
	W\equiv W[a,\alpha ]=[w_{1}w_{2}\ldots w_{d_{\alpha }}]\in \M(d_{\alpha }\times
	\dim M_{f}^{\alpha },\mathbb{C}) 
	\ee
	has the columns which are eigenvectors for eigenvalue $d$ of the matrix $M_{f}^{\alpha }(V'[(a,n)])$, defined in Proposition~\ref{eigen1}.
\end{definition}

Next we will need also the dimension structure of the natural
representation of $\A$. Now we recall the

\begin{theorem}
	The algebra $A_{n}^{t_{n}}(d)$ in its natural representation in the space $(\mathbb{C}^{d})^{\otimes n}$ has the following decomposition into irreps
	\be
	V_{d}[A_{n}^{t_{n}}(d)]=\left[ \bigoplus _{\alpha :h(\alpha )\leq d}m_{\alpha }\Phi
	^{\alpha }\right] \oplus \left[\bigoplus _{\nu :h(\nu )<d}M_{\nu }\Psi ^{\nu }\right] , 
	\ee
	where the multiplicity $m_{\alpha }$ is equal to the multiplicity of the irrep $\varphi ^{\alpha }$ of $S(n-2)$ in the representation $V_{d}[S(n-2)]$ i.e.
	\be
	m_{\alpha }=\frac{1}{(n-2)!}\sum_{\sigma \in S(n-2)}\chi ^{\alpha }(\sigma
	^{-1})d^{l(\sigma )} 
	\ee
	and 
	\be
	M_{\nu }=dm_{\nu }-\sum_{\alpha :\nu \in \ind_{S(n-2)}^{S(n-1)}(\varphi
		^{\alpha })}m_{\alpha }. 
	\ee
\end{theorem}

From the above theorem we deduce

\begin{corollary}
	\label{cor7}
	We have the following relation between dimensions of the natural
	representation space and dimensions and multiplicities of the irreps of the algebra $A_{n}^{t_{n}}(d)$ 
	\be
	d^{n}=\sum_{\alpha :h(\alpha )\leq d}m_{\alpha }\dim \Phi ^{\alpha
	}+\sum_{_{\nu :h(\nu )<d}}M_{\nu }\dim \psi ^{\nu }, 
	\ee
	or equivalently
	\be
	1=\frac{1}{d^{n}}\sum_{\alpha :h(\alpha )\leq d}m_{\alpha }\dim \Phi
	^{\alpha }+\frac{1}{d^{n}}\sum_{_{\nu :h(\nu )<d}}M_{\nu }\dim \psi ^{\nu }. 
	\ee
\end{corollary}

Further in the proof of the Theorem~\ref{limits} we will need also very classical inequality. Namely we have the following:

\begin{theorem}
	\label{jensen}
	Suppose that the function $f:\mathbb{R}\rightarrow \mathbb{R}$ is convex in some subset $D\subset \mathbb{R}$ of its domain, then for any numbers $\lambda _{1},\ldots,\lambda _{m}\in D$
	and for any probability distribution $\sum_{i=1}^{m}s_{i}=1$ we have 
	\be
	f\left( \sum_{i=1}^{m}s_{i}\lambda _{i}\right) \leq \sum_{i=1}^{m}s_{i}f(\lambda _{i}). 
	\ee
	Moreover it is known, that the function $f:\mathbb{R} \rightarrow \mathbb{R}$, which is twice differentiable on a subset $D\subset 
    \mathbb{R}$ and 
	\be
	f^{\prime \prime }(\lambda )\geq 0,\qquad \lambda \in D 
	\ee
	is convex in subset $D\subset \mathbb{R}$ of its domain.
\end{theorem}

From Theorem~\ref{jensen} we can deduce that

\begin{proposition}
	\label{prop49}
	The function $f(\lambda )=\lambda ^{-1/2}:\lambda \geq 0$ is convex, because $f^{\prime \prime }(\lambda )=\frac{3}{4}\lambda
	^{-5/2}\geq 0$.
\end{proposition}

Now we are in the position to prove Theorem~\ref{limits}. The proof is the following.\\

\begin{proof}[Proof of Theorem~\ref{limits}]
	First we prove the bound condition. Using the spectral decomposition 
	\be
	M_{f}^{\alpha }(V'[(a,n)])=d\sum_{i_{\alpha }=1}^{d_{\alpha
	}}w_{i_{\alpha }}[a,\alpha ]w_{i_{\alpha }}^{\dagger}[a,\alpha ] 
	\ee
	by a direct calculation one gets
	\be
	\tr M_{f}^{\alpha }\left[ \rho ^{-\frac{1}{2}}V'[(a,n)]\rho ^{-\frac{1}{2}
	}V'[(a,n)]\right] =d^{2}\sum_{i_{\alpha },j_{\alpha }=1}^{d_{\alpha
	}}\left| (w_{i_{\alpha }},M_{f}^{\alpha }(\rho ^{-\frac{1}{2}})w_{j_{\alpha
	}})\right| ^{2}.
	\ee
	On the other hand similarly we have 
	\be
	\tr\left[ W^{\dagger}M_{f}^{\alpha }(\rho ^{-\frac{1}{2}})W\right] =\sum_{p=d_{\alpha }+1}^{\dim
		M_{f}^{\alpha }}\sum_{i_{\alpha }=1}^{d_{\alpha }}\left| (w_{i_{\alpha
	}},M_{f}^{\alpha }(\rho ^{-\frac{1}{2}})w_{j_{\alpha
	}})\right| ^{2}+\sum_{i_{\alpha },j_{\alpha }=1}^{d_{\alpha }}\left| (w_{i_{\alpha
	}},M_{f}^{\alpha }(\rho ^{-\frac{1}{2}})w_{j_{\alpha }})\right| ^{2},
	\ee
	where the rectangular matrix $W$ is described in Definition~\ref{def5}. From these
	equations we deduce 
	\be
	\tr\left[ W^{\dagger}M_{f}^{\alpha }(\rho ^{-1})W\right] =\sum_{p=d_{\alpha }+1}^{\dim
		M_{f}^{\alpha }}\sum_{i_{\alpha }=1}^{d_{\alpha }}\left| (w_{i_{\alpha
	}},M_{f}^{\alpha }(\rho ^{-\frac{1}{2}})w_{j_{\alpha }})\right| ^{2}+\frac{1}{d^{2}}%
	\tr M_{f}^{\alpha }[\rho ^{-\frac{1}{2}}V'[(a,n)]\rho ^{-\frac{1}{2}%
	}V'[(a,n)]],
	\ee
	so 
	\be
	d^{2}\tr\left[ W^{\dagger}M_{f}^{\alpha }(\rho ^{-1})W\right] \geq \tr M_{f}^{\alpha }[\rho ^{-
		\frac{1}{2}}V'[(a,n)]\rho ^{-\frac{1}{2}}V'[(a,n)]].
	\ee
	Now using the explicit form of the matrix $W$ given in Definition~\ref{def5} and Proposition~\ref{eigen1} we
	get the following formula 
	\be
	\tr\left[ W^{\dagger}M_{f}^{\alpha }(\rho ^{-1})W\right] =\frac{1}{(n-1)d}\dim M_{f}^{\alpha },
	\ee
	Finally we get the following upper bound for trace of the operator $\rho ^{-\frac{1}{2}}V'[(a,n)]\rho ^{-\frac{1}{2}}V'[(a,n)]$ in the irrep $M_{f}^{\alpha }$ 
	\be
	\frac{d}{(n-1)}\dim M_{f}^{\alpha }\geq \tr M_{f}^{\alpha }[\rho ^{-\frac{1}{2}}V'[(a,n)]\rho ^{-\frac{1}{2}}V'[(a,n)]].
	\ee
	From this we deduce the upper bound for fidelity $F_{n}(d)$
	in the following way
	\be
	\begin{split}
	F_n(d)&=\frac{n-1}{d^{n+1}}\tr_{(\mathbb{C}^{d})^{\otimes n}}[\rho ^{-\frac{1}{2}}V'[(a,n)]\rho ^{-\frac{1}{2}
	}V'[(a,n)]]\\
	&=\frac{n-1}{d^{n+1}}\sum_{\alpha :h(\alpha )\leq d}m_{\alpha
	}\tr M_{f}^{\alpha }[\rho ^{-\frac{1}{2}}V'[(a,n)]\rho ^{-\frac{1}{2}%
	}V'[(a,n)]]\\\
	&\leq \frac{n-1}{d^{n+1}}\sum_{\alpha :h(\alpha )\leq d}m_{\alpha }\frac{d}{%
		(n-1)}\dim M_{f}^{\alpha }=\frac{1}{d^{n}}\sum_{\alpha :h(\alpha )\leq
		d}m_{\alpha }\dim M_{f}^{\alpha }.
	\end{split}
	\ee
	Now from Corollary~\ref{cor7} we have
	\be
	\frac{1}{d^{n}}\sum_{\alpha :h(\alpha )\leq d}m_{\alpha }\dim M_{f}^{\alpha
	}=1-\frac{1}{d^{n}}\sum_{_{\nu :h(\nu )<d}}M_{\nu }\dim \psi ^{\nu }<1,
	\ee
	so in this way we get the first statement of Theorem~\ref{limits}
	\be
	F_n(d)\leq 1-\frac{1}{d^{n}}\sum_{_{\nu :h(\nu )<d}}M_{\nu
	}\dim \psi ^{\nu }<1.
	\ee
	In order to prove the remaining part of the theorem we consider the
	generalised fidelity 
	\be
	F_n(d)=\frac{n-1}{d^{n+1}}\tr_{(\mathbb{C}^{d})^{\otimes n}}[\rho  ^{-\frac{1}{2}}V'[(a,n)]\rho ^{-\frac{1}{2}}V'[(a,n)]]=\frac{n-1}{%
		d^{n+1}}\sum_{\alpha :h(\alpha )\leq d}\frac{d_{\alpha }}{m_{\alpha }}
	\left[ \sum_{\mu \neq \theta }\lambda ^{-\frac{1}{2}} _{\mu}(\alpha)m_{\mu }\right] ^{2},
	\ee
	which may rewritten as follows
	\be
	F_{n}(d)=\frac{1}{(n-1)d^{n+1}}\sum_{\alpha :h(\alpha )\leq d}\frac{%
		m_{\alpha }}{d_{\alpha }}\left[ \sum_{\mu \neq \theta }\lambda _{\mu} ^{\frac{1}{2}}(\alpha)d_{\mu }\right] ^{2}.
	\ee
	Further
	\be
	\begin{split}
	F_n(d)&=\frac{1}{(n-1)d^{n+1}}\sum_{\alpha :h(\alpha )\leq d}\frac{%
		m_{\alpha }}{d_{\alpha }}[d(n-1)d_{\alpha }]^{2}\left[ \sum_{\mu \neq \theta }%
	\frac{d_{\mu }\lambda _{\mu}(\alpha)}{d(n-1)d_{\alpha }}\lambda _{\mu} ^{-\frac{1}{2}}(\alpha)\right] ^{2}\\
	&=\frac{n-1}{d^{n-1}}\sum_{\alpha :h(\alpha )\leq d}m_{\alpha }d_{\alpha
	}\left[ \sum_{\mu \neq \theta }s_{\mu }\lambda _{\mu} ^{-\frac{1}{2}}(\alpha)\right] ^{2},
\end{split}
	\ee
	where 
	\be
	s_{\mu }=\frac{d_{\mu }\lambda _{\mu}(\alpha)}{d(n-1)d_{\alpha }} \qquad \text{with} \qquad \sum_{\mu \neq \theta }s_{\mu }=\sum_{\mu =\alpha +\square
	}s_{\mu }=1,
	\ee
	because $\sum_{\mu \neq \theta }d_{\mu }\lambda _{\mu}(\alpha)=\sum_{\mu
		=\alpha +\square }d_{\mu }\lambda _{\mu}(\alpha)=\tr Q(\alpha
	)=d(n-1)d_{\alpha },$ thus the sum 
	\be
	\sum_{\mu \neq \theta }s_{\mu }\lambda _{\mu} ^{-\frac{1}{2}}(\alpha)
	\ee
	is a convex combination of the numbers $\lambda _{\mu ,\alpha }^{-1/2}$, and we
	may use Theorem~\ref{jensen}, and Proposition~\ref{prop49}, which give 
	\be
	\left( \sum_{\mu \neq \theta }s_{\mu }\lambda
	_{\mu}(\alpha)\right) ^{-\frac{1}{2}}\leq \sum_{\mu \neq \theta }s_{\mu }\lambda _{\mu} ^{-\frac{1}{2}}(\alpha).
	\ee
	Using above inequality we get the following lower bound for generalised
	fidelity%
	\be
	\begin{split}\label{fidelitylb}
	F_{n}(d)&=\frac{n-1}{d^{n-1}}\sum_{\alpha :h(\alpha )\leq d}m_{\alpha
	}d_{\alpha }\left(  \sum_{\mu \neq \theta }s_{\mu }\lambda _{\mu}^{-\frac{1}{2}}(\alpha)\right) ^{2}
	\geq \frac{n-1}{d^{n-1}}\sum_{\alpha :h(\alpha )\leq d}m_{\alpha }d_{\alpha
	}\left( \sum_{\mu \neq \theta }s_{\mu }\lambda _{\mu}(\alpha)\right) ^{-1}\\
    &=\frac{n-1}{%
		d^{n-1}}\sum_{\alpha :h(\alpha )\leq d}m_{\alpha }d_{\alpha }\left( \sum_{\mu \neq
		\theta }\frac{d_{\mu }\lambda _{\mu}^{2}(\alpha)}{d(n-1)d_{\alpha }}\right) ^{-1}
	=\frac{n-1}{d^{n-1}}\sum_{\alpha :h(\alpha )\leq d}m_{\alpha }d_{\alpha }\left( 
	\frac{\tr Q^{2}(\alpha )}{\tr Q(\alpha )}\right) ^{-1}.
	\end{split}
	\ee
	Directly from Definition~\ref{def_Q} of the matrix $Q(\alpha )$ we get that $%
	\tr Q^{2}(\alpha )=(n-1)d_{\alpha }(d^{2}+n-2)$, and
	\be
	\label{gen_bound}
	F_{n}(d)\geq \frac{n-1}{d^{n-1}}\sum_{\alpha :h(\alpha )\leq d}m_{\alpha
	}d_{\alpha }\left( \frac{d^{2}+n-2}{d}\right) ^{-1}=\frac{n-1}{d}\left( \frac{d^{2}+n-2}{d}\right)^{-1},
	\ee
    so 
	\be
	\label{Koning}
	F_{n}(d)\geq \frac{n-1}{d^{2}+n-2},
	\ee
	and thus recover the result of~\cite{beigi_konig} obtained using different method.
	Inequality~\eqref{Koning} together with the upper bound $F_{n}(d)<1$ implies that 
	\be
	\lim_{n\rightarrow \infty }F_{n}(d)=1.
	\ee
\end{proof}
\section{Discussion and open problems}
We found significant simplifications of the algebra $\A$  of the partially transposed permutation operators with respect to last subsystem  by developing tools of PRIRs by proving a few new orthogonality theorems for them. Our successful approach to study PRIRs we apply to algebra $\A$  by simplifying existing theorems. The main simplifications concern to matrix $Q$ given in Definition~\ref{def_Q} and matrix $Z$ given in Theorem~\ref{thm18} constructed from eigenvectors of $Q$. We were able to reduce complexity of underlying expressions by reducing number of sums over all permutation from $S(n)$. Such a reduction allows us to perform any calculations especially devoted to practical applications discussed later more efficiently. Second main result obtained thanks to new approach are relatively simple equations for the matrix elements of operators $V'[(a,n)]$ for $a=1,\ldots,n$ with particular case when $a=n-1$.
 
Finally we applied derived simplifications to obtain characteristic of the deterministic port-based teleportation scheme. Firstly we gave explicit connection between PBT operator $\rho$ and matrix $Q$ describing properties of the algebra $\A$. We have shown that non-zero eigenvalues of $Q$ are exactly eigenvalues of the operator $\rho$. Later we presented derivation for the fidelity $F$ of the teleported state and expressed final result by parameters describing irreps of $\A$ like dimensions and multiplicities as well as  global parameters $d$ and $N$. We presented asymptotic analysis  of $F$ showing that $\operatorname{lim}_{N\rightarrow \infty}F=1$ for fixed $d$ which certifies our approach. Moreover using completely new method of computation based on analysis of the eigen-problem of $V'[(a,n)]$ for $a=1,\ldots,n-1$ we derived known non-trivial lower bound for the fidelity $F$ expressed only by global parameters as number of ports $N$ and dimension $d$.

Despite of progress made in this manuscript and papers~\cite{Stu2017,Moz2017b} there are still a few open questions connected with the general theory of algebra $\A$ and possible applications to PBT. The most interesting and important question in opinion of the authors would be a full  eigen-analysis of the PBT operator $\rho$ similarly as for $V'[(a,n)]$ in Proposition~\ref{eigen1}. Namely we would like to find its in  eigenvectors (since eigenvalues are known) in terms of parameters describing irreps of $\A$ and analyse its entanglement with respect to some particular cuts. Such a analysis would be helpful in extension of the hybrid PBT to higher dimensions. Additionally simplification presented in this paper should give technically easier from the perspective of representation theory description of $1\rightarrow N$ universal quantum cloning machines~\cite{Stu2014}.

\section*{Acknowledgements}
MS~is supported by the grant "Mobilno{\'s}{\'c} Plus IV", 1271/MOB/IV/2015/0 from the Polish Ministry of Science and Higher Education. MH and MM are supported by National Science Centre, Poland, grant OPUS 9.  2015/17/B/ST2/01945. 
\appendix
\section{Summary of known fact about PRIRs}
\label{AppAA}
In this appendix for self-consistence of our manuscript we present known fact about PRIRs. For all proofs we refer reader to~\cite{Stu2017}. 
\begin{proposition}
	Let $(\psi _{R}^{\mu })^{\alpha \alpha }(\sigma )=\left( \psi _{i_{\alpha
		}j_{\alpha }}^{\alpha \alpha }(\sigma )\right) $ be the matrices on the diagonal of
	the PRIR matrix $\psi _{R}^{\mu }(\sigma )$ where $\sigma \in S(n)$, then%
	\be
	\forall \alpha \in \mu \quad \varphi ^{\alpha }(\pi )\left( \psi _{R}^{\mu
	})^{\alpha \alpha }[(a,n)]\right) \varphi ^{\alpha }(\pi^{-1})=\left( \psi _{R}^{\mu
	})^{\alpha \alpha }[(\pi(a),n)]\right) , 
	\ee
	and from this it follows%
	\be
	\forall \alpha \in \mu \quad \forall \pi \in S(n-1)\quad \forall
	a=1,\ldots,n-1\qquad \tr\left[ (\psi _{R}^{\mu })^{\alpha \alpha }[(a,n)]\right] =\tr\left[ (\psi _{R}^{\mu
	})^{\alpha \alpha }[(\pi (a),n)]\right] , 
	\ee
	so the trace in each diagonal block is constant  on the transpositions which naturally indexed the coset $S(n)/S(n-1).$
\end{proposition}

\begin{proposition}
	\label{Prir1}
	The PRIR $\psi _{R}^{\mu }$ of $S(n)$ satisfies the following summation rules%
	\be
	\label{Prir1eq}
	\sum_{a=1}^{n-1}(\psi _{R}^{\mu })[(a,n)]=\frac{n(n-1)}{2}\frac{\chi ^{\mu
		}[(1,2)]}{d_{\mu }}\mathbf{1}_{\psi ^{\mu }}-\bigoplus _{\alpha \in \mu }\frac{%
		(n-1)(n-2)}{2}\frac{\chi ^{\alpha }[(1,2)]}{d_{\alpha }}\mathbf{1}_{\varphi
		^{\alpha }}, 
	\ee
	which implies that for the diagonal blocks we have 
	\be
	\forall \alpha \in \mu \qquad \sum_{a=1}^{n-1}(\psi _{R}^{\mu })^{\alpha
		\alpha }[(a,n)]=\left[ \frac{n(n-1)}{2}\frac{\chi ^{\mu }[(1,2)]}{d_{\mu }}
	-\frac{(n-1)(n-2)}{2}\frac{\chi ^{\alpha }[(1,2)]}{d_{\alpha }}%
	\right] \mathbf{1}_{\varphi_{\alpha }}. 
	\ee
\end{proposition}

\begin{remark}
	Equation~\eqref{Prir1eq} in Proposition~\ref{Prir1} may be written in a  more
	explicit form as follows:
	\be
	\forall \alpha \in \mu \qquad \sum_{a=1}^{n-1}(\psi _{R}^{\mu })_{i_{\alpha
		}j_{\alpha }}^{\alpha \alpha }[(a,n)]=\left[ \frac{n(n-1)}{2}\frac{\chi ^{\mu }[(1,2)]%
	}{d_{\mu }}-\frac{(n-1)(n-2)}{2}\frac{\chi ^{\alpha }[(1,2)]}{d_{\alpha }}\right] \delta _{i_{\alpha }j_{\alpha }},
	\ee
	where $i_{\alpha},j_{\alpha}=1,\ldots,d_{\alpha}$.
\end{remark}

\bibliographystyle{plain}
\bibliography{biblio}
\end{document}